\documentclass[11pt,a4paper]{article}
\pdfoutput=1

\usepackage{amssymb}
\usepackage{amsmath}
\usepackage{amsfonts}
\usepackage{bbm}
\usepackage{amsthm}
\usepackage{mathrsfs}
\usepackage{hyperref}
\usepackage{color}
\usepackage[margin=2.35cm]{geometry}
\usepackage[all,cmtip]{xy}
\usepackage[utf8]{inputenc}
\usepackage{graphicx}
\usepackage{varwidth}
\usepackage{comment}
\usepackage{enumitem}

\usepackage{mathtools}

\usepackage{upgreek}
\usepackage{rotating}

\usepackage{tikz}
\usetikzlibrary{shapes.geometric}

\usetikzlibrary{decorations.markings}

\usepackage{tikz-cd}
\tikzset{Rightarrow/.style={double equal sign distance,>={Implies},->},
triple/.style={-,preaction={draw,Rightarrow}},
quadruple/.style={preaction={draw,Rightarrow,shorten >=0pt},shorten >=1pt,-,double,double
distance=0.2pt}}

\mathchardef\mhyphen="2D

\definecolor{darkred}{rgb}{0.8,0.1,0.1}
\hypersetup{
	colorlinks=true,         
	linkcolor=darkred,
	citecolor=blue,
}

\theoremstyle{plain}
\newtheorem{theo}{Theorem}[section]
\newtheorem{lem}[theo]{Lemma}
\newtheorem{propo}[theo]{Proposition}

\theoremstyle{definition}
\newtheorem{defi}[theo]{Definition}

\newtheorem{exx}[theo]{Example}
\newenvironment{ex}
{\pushQED{\qed}\exx}
{\popQED\endexx}

\newtheorem{remm}[theo]{Remark}
\newenvironment{rem}
{\pushQED{\qed}\remm}
{\popQED\endremm}

\numberwithin{equation}{section}

\def\nn{\nonumber}

\def\bbK{\mathbb{K}}
\def\bbR{\mathbb{R}}

\def\bbN{\mathbb{N}}
\def\bbZ{\mathbb{Z}}

\def\hom{\underline{\mathrm{hom}}}

\def\Sym{\mathrm{Sym}}

\def\id{\mathrm{id}}

\def\supp{\mathrm{supp}}

\def\dd{\mathrm{d}}

\def\dR{\mathrm{dR}}

\def\1{I}
\def\oone{\mathbbm{1}}
\def\op{\mathrm{op}}

\def\pr{\mathrm{pr}}

\def\Set{\mathbf{Set}}

\def\CAlg{\mathbf{CAlg}}

\def\dgCAlg{\mathbf{dgCAlg}}

\def\dgMod{\mathbf{dgMod}}

\def\sSet{\mathbf{sSet}}

\def\dgCat{\mathbf{dgCat}}

\def\Fun{\mathbf{Fun}}

\def\2Cat{\mathbf{2Cat}}
\def\3Cat{\mathbf{3Cat}}

\def\Aff{\mathbf{Aff}}
\def\dAff{\mathbf{dAff}}
\def\dPSt{\mathbf{dPSt}}
\def\dSt{\mathbf{dSt}}
\def\Rect{\mathbf{Rect}}

\def\Spec{\mathrm{Spec}}

\def\Ad{\mathrm{Ad}}

\def\g{\mathfrak{g}}

\def\FFF{\mathfrak{F}}

\def\G{\mathcal{G}}

\def\O{\mathcal{O}}
\def\P{\mathcal{P}}

\def\dCrit{\mathrm{dCrit}}
\def\Sol{\mathcal{S}}

\def\holim{\mathrm{holim}}

\def\hocolim{\mathrm{hocolim}}

\def\gf{\mathrm{gf}}

\def\GL{\mathrm{GL}}
\def\gl{\mathfrak{gl}}
\def\Con{\mathrm{Con}}

\def\Perf{\mathrm{Perf}}
\def\QCoh{\mathrm{QCoh}}

\newcommand\und[1]{\underline{#1}}

\makeatletter
\newcommand{\xRrightarrow}[2][]{\ext@arrow 0359\Rrightarrowfill@{#1}{#2}}
\newcommand{\Rrightarrowfill@}{\arrowfill@\equiv\equiv\Rrightarrow}
\newcommand{\xLleftarrow}[2][]{\ext@arrow 3095\Lleftarrowfill@{#1}{#2}}
\newcommand{\Lleftarrowfill@}{\arrowfill@\Lleftarrow\equiv\equiv}
\makeatother

\def\sk{\vspace{2mm}}

\makeatletter
\let\@fnsymbol\@alph
\makeatother

\title{%
Derived algebraic geometry of 2d lattice Yang-Mills theory
}

\author{%
Marco Benini$^{1,2,a}$, Tom{\'a}s Fern{\'a}ndez$^{3,b}$\ and\ Alexander Schenkel$^{4,5,c}$\vspace{4mm}\\
{\small ${}^1$ Dipartimento di Matematica, Dipartimento di Eccellenza 2023-27, Universit\`a di Genova,}\\
{\small Via Dodecaneso 35, 16146 Genova, Italy.}\vspace{2mm}\\
{\small ${}^2$ INFN, Sezione di Genova,}\\
{\small Via Dodecaneso 33, 16146 Genova, Italy.}\vspace{2mm}\\
{\small ${}^3$~{\'E}cole Normale Sup{\'e}rieure de Lyon,}\\
{\small 15 parvis Ren{\'e} Descartes, 69342 Lyon Cedex 07, France.}\vspace{2mm}\\
{\small ${}^4$~Dipartimento di Matematica, Universit{\`a} di Trento and INFN-TIFPA,}\\
{\small Via Sommarive 14, 38123 Povo (Trento), Italy.}\vspace{2mm}\\
{\small ${}^5$~School of Mathematical Sciences, University of Nottingham,}\\
{\small University Park, Nottingham NG7 2RD, United Kingdom.}\vspace{4mm}\\
{\small
Email: ${}^a$~\href{mailto:marco.benini@unige.it}{\texttt{marco.benini@unige.it}}, 
${}^b$~\href{mailto:tomasfp@outlook.com}{\texttt{tomasfp@outlook.com}},
${}^c$~\href{mailto:alexander.schenkel@unitn.it}{\texttt{alexander.schenkel@unitn.it}}
}
}

\date{June 2026}

\begin{document}

\maketitle

\begin{abstract}
\noindent A derived algebraic geometric study of classical $\mathrm{GL}_n$-Yang-Mills theory on the $2$-dimensional square lattice $\mathbb{Z}^2$ is presented. The derived critical locus of the Wilson action is described and its local data supported in rectangular subsets $V =[a,b]\times [c,d]\subseteq \mathbb{Z}^2$ with both sides of length $\geq 2$ is extracted. A locally constant dg-category-valued prefactorization algebra on $\mathbb{Z}^2$ is constructed from the dg-categories of quasi-coherent complexes on the derived stacks of local data.
\end{abstract}
\vspace{-3mm}

\vspace{-2mm}

\vspace{-2mm}

\renewcommand{\baselinestretch}{0.8}\normalsize
\tableofcontents
\renewcommand{\baselinestretch}{1.0}\normalsize


\section{\label{sec:intro}Introduction and summary}
Derived algebraic geometry is a powerful refinement
of algebraic geometry in which
one can give a precise geometric meaning to objects
that are problematic in traditional approaches, 
such as quotients by non-free group actions and non-transversal intersections.
This is achieved by simultaneously considering $\infty$-groupoids instead of sets of points
and enlarging the basic building blocks from affine schemes to derived affine schemes, 
whose algebras of functions are allowed to 
carry additional homotopical structure, taking the 
form of commutative dg-algebra structures in non-positive 
cohomological degrees when working over a field of characteristic zero.
The combination of these two refinements makes such singular objects behave smoothly.
The resulting spaces which are obtained by gluing 
derived affine schemes are known as derived stacks.
We refer the reader to \cite{ToenVezzosi,GRbook} for the 
foundations of derived algebraic geometry
and also to \cite{Toen,Calaque,Pridham} for more concise introductions.
\sk

In addition to its intrinsic relevance to the foundations of algebraic geometry,
derived algebraic geometry also has an ever-increasing impact
on other disciplines. For instance, modern approaches to (quantum) field theory,
such as the factorization algebras of Costello and Gwilliam \cite{CG1,CG2}, 
are heavily inspired by the ideas and techniques of derived algebraic geometry.
The relationship between field theory and derived geometry can be seen most directly 
as follows: The main object of interest in a classical field theory is the moduli space
of solutions to some system of partial differential equations
which is usually given by the Euler-Lagrange equations of some action function $S$.
Such moduli spaces are described by the intersection problem $\delta S = 0$
associated with the variation of the action, which is called a derived critical locus. 
These moduli spaces carry in general a non-trivial derived geometric structure
which results from non-transversality of the intersection problem
and non-freeness of the action of gauge symmetries in gauge field theories.
\sk

It is worthwhile to note that most of the current applications
of derived geometry to (quantum) field theory are intrinsically 
perturbative. This means that they do not attempt to 
describe the derived stack encoding the global moduli space
of the field theory, but they focus only on the formal neighborhood 
of a point (i.e.\ a formal moduli problem), interpreted as a background solution around which 
one considers formal perturbations. The main reason for this limitation
is that field theories do not strictly fit into the
standard framework of derived algebraic geometry since their description 
requires functional analytical objects, such as the spaces 
of smooth or distributional sections of vector bundles over a manifold,
which lie outside the scope of this approach. 
Very recently there has been substantial progress towards
generalizing the framework of derived algebraic geometry such 
that it becomes applicable to partial differential equations and 
other kinds of functional analytical objects. Some notable developments 
are the works of Steffens \cite{Steffens1,Steffens2} on 
derived $C^\infty$-algebraic geometry, the work of 
Ben-Bassat, Kelly and Kremnizer \cite{Analytic} on 
derived analytic geometry, and the works of 
Kryczka, Sheshmani and Yau \cite{DerivedPDE1,DerivedPDE2}
on a $D$-module approach to the derived geometry of partial differential equations.
These novel frameworks are however highly technical and abstract,
such that their application to concrete questions in (quantum) field theory
remains an open problem for future works.
\sk

In this work we take a complementary approach which is inspired by
lattice field theory, see e.g.\ \cite{Montvay}. The basic idea is to approximate the underlying
spacetime manifold of a field theory by some discrete structure, 
such as a square lattice $\bbZ^n$. This removes the need for functional
analytical objects such as distribution spaces and it replaces the
partial differential equations from continuum field theory 
with simpler finite-difference equations. The moduli spaces 
associated with such systems thus lie within the scope of 
standard derived algebraic geometry.
The main aim of this paper is to show
that non-perturbative lattice field theories can be described and studied
rather explicitly using the methods of derived algebraic geometry.
For concreteness, we shall focus on the example of classical
$\GL_n$-Yang-Mills theory on the $2$-dimensional square lattice $\bbZ^2$, 
which displays all the main features of interest to us, namely 
a highly non-linear dynamics and a non-Abelian gauge group.
The moduli space for this model is defined as the derived critical locus
of the Wilson action \cite{Wilson}, which is a discretization of the 
Yang-Mills action from continuum gauge theory. 
\sk

We will now explain our results by outlining the content of this paper.
In Section \ref{sec:prelim} we collect some relevant preliminaries.
In Subsection \ref{subsec:DAG} we recall those concepts
of derived algebraic geometry which are necessary for our work, 
including derived affine schemes, derived (quotient) stacks and their dg-categories of quasi-coherent complexes.
In Subsection \ref{subsec:dCrit} we recall the concept of a derived critical
locus and its explicit description from \cite{BSSdCrit} for the
case of a function $S : [X/G]\to \mathbb{A}^1$ on a quotient stack.
Subsection \ref{subsec:latticeYM} provides a very brief introduction to lattice
gauge theory and in particular recalls the Wilson action \cite{Wilson} on $\bbZ^2$ as well as
its Euler-Lagrange equations, which are highly non-linear finite-difference equations.
\sk

In Section \ref{sec:dCrit} we provide a very explicit description
of the global derived critical locus of the Wilson action on $\bbZ^2$
in terms of a derived quotient stack  $\dCrit(S)\simeq \big[Z(\bbZ^2)/\G(\bbZ^2)\big]$,
see in particular \eqref{eqn:OGau}, \eqref{eqn:OZYM} and \eqref{eqn:rhodCrit}.
We will prove in Subsection \ref{subsec:axialgauge} that this derived stack
admits a weakly equivalent description implementing an axial gauge fixing condition,
i.e.\ fixing one of the two components of the connection on $\bbZ^2$ to be trivial.
This turns out to be very useful for understanding the dynamics 
of the lattice Yang-Mills model. In Section \ref{sec:localconstancy}
we extract from the global derived critical locus $\dCrit(S)$ on $\bbZ^2$ the local
data which is supported in rectangular subsets $V = [a,b]\times[c,d]\subseteq \bbZ^2$
with both sides of length $\geq 2$. This defines a functor
$\Sol:\Rect(\bbZ^2)^\op\to\dSt$ from the opposite of the category $\Rect(\bbZ^2)$
of such rectangular subsets and their inclusions to the model category of derived stacks.
The main result of this section is Theorem \ref{theo:locallyconstant} 
in which we prove that this functor is locally constant in the sense that the restriction map
$\Sol(V^\prime)\to\Sol(V)$ is a weak equivalence of derived stacks
for every inclusion  $V\subseteq V^\prime$ of rectangular subsets.
This is a non-trivial and rather technical result which verifies the physical intuition that
``classical $2$-dimensional (lattice) Yang-Mills theory does not contain local propagating degrees of freedom''.
(It is important to stress that this feature holds true only for the classical theory as the 
corresponding quantum theory is area-dependent, see e.g.\ \cite{Witten}.)
\sk

In Section \ref{sec:PFA} we connect our constructions and results to (pre)factorization algebras.
We show that our non-perturbative classical lattice Yang-Mills model defines
a discrete variant of a prefactorization algebra on $\bbZ^2$  
taking values in the $2$-category of dg-categories,
see Definition \ref{def:PFA} for the relevant prefactorization operad.
The reason for the appearance of dg-categories, in contrast to
cochain complexes as for perturbative (pre)factorization algebras
\cite{CG1,CG2}, is that our derived stacks of local data $\Sol(V)$ 
are not affine, i.e.\ they are not faithfully encoded by their dg-algebras
of functions. Instead, our derived stacks $\Sol(V) = \big[Z(V)/\G(V)\big]$ are
quotients satisfying the finiteness conditions 
from \cite[Theorem 2.2.4]{Gaitsgory}, so they are $1$-affine and therefore determined
by their symmetric monoidal dg-category of quasi-coherent complexes
$\QCoh\big(\Sol(V)\big)$. 
(As a consequence of local constancy from Theorem \ref{theo:locallyconstant}, 
this is also true for unbounded rectangular subsets.)
The main result of this section is that this dg-category-valued prefactorization
algebra on $\bbZ^2$ is locally constant with respect to quasi-equivalences
of dg-categories, see Theorem \ref{theo:PFA}.
We would like to emphasize that our prefactorization algebra
describes only the classical non-perturbative observables of
$2$-dimensional lattice Yang-Mills theory and that its quantization
remains an open problem, see Remark \ref{rem:quantization} for further comments.
Theorem \ref{theo:PFA} also creates links between our work and the 
``not too little disks'' algebras which have been
developed recently by Calaque and Carmona \cite{CalaqueCarmona},
see Remark \ref{rem:nottoolittle} for further comments.


\section{\label{sec:prelim}Preliminaries}

\subsection{\label{subsec:DAG}Basic derived algebraic geometry}
We will briefly recall some relevant concepts of derived algebraic geometry
which are needed to state and prove the results of our work. We refer
the reader to \cite{ToenVezzosi,GRbook} for details and to
\cite{Toen,Calaque,Pridham} for more concise introductions.
Let us fix once and for all a field $\bbK$ of characteristic $0$.
\sk

The basic objects on which derived algebraic geometry is built
are derived affine schemes, providing
a homological refinement of the ordinary affine schemes from algebraic geometry.
\begin{defi}\label{def:dAff}
The category of \textit{derived affine schemes}
\begin{flalign}
\dAff\,:=\, \big(\dgCAlg^{\leq 0}\big)^\op 
\end{flalign}
is defined as the opposite of the category $\dgCAlg^{\leq 0}$ of 
commutative dg-algebras over $\bbK$
in non-positive cohomological degrees. This means that a morphism 
$f : \Spec(A)\to \Spec(B)$ in $\dAff$ is defined 
by an opposite morphism $f^\ast : B\to A$ in $\dgCAlg^{\leq 0}$.
We endow $\dAff$ with the opposite
of the standard model structure on $\dgCAlg^{\leq 0}$ (see e.g.\ \cite{CDGA}),
i.e.\ a morphism $f : \Spec(A)\to \Spec(B)$ in $\dAff$ is
\begin{itemize}
\item a weak equivalence if its opposite $f^\ast : B\to A$ is a quasi-isomorphism,
\item a cofibration if its opposite $f^\ast: B\to A$ is surjective in all negative degrees $<0$, and
\item a fibration if it has the right-lifting property with respect to all morphisms
that are both a weak equivalence and a cofibration.
\end{itemize}
\end{defi}
\begin{rem}
The evident embedding $\CAlg\to \dgCAlg^{\leq 0}$ of the category 
of commutative $\bbK$-algebras defines an embedding 
\begin{flalign}
\Aff~\longrightarrow~ \dAff
\end{flalign}
of the category of ordinary affine schemes $\Aff:= \CAlg^\op$ into 
the model category of derived affine schemes from Definition \ref{def:dAff}.
Hence, every ordinary affine scheme gives rise to a derived affine scheme.
\end{rem}

Derived affine schemes are insufficient to describe certain important 
geometric objects, such as quotients.
This issue can be resolved by enlarging the model category $\dAff$ from Definition \ref{def:dAff}
to the model category $\dSt$ of derived stacks from \cite{ToenVezzosi}. Loosely speaking,
a derived prestack is a simplicial presheaf $X : \dAff^\op \to \sSet$
which sends weak equivalences in $\dAff^\op = \dgCAlg^{\leq 0}$ 
to weak equivalences in the Kan-Quillen model structure on the category of simplicial sets $\sSet$.
A derived stack is a derived prestack which additionally 
satisfies hyperdescent with respect to {\'e}tale hypercovers of derived affines.
\begin{defi}\label{def:dSt}
The model category of \textit{derived prestacks}
\begin{flalign}
\dPSt\,:=\, \mathcal{L}_{\widehat{W}}\,\mathbf{sPSh}(\dAff)
\end{flalign}
is defined as the left Bousfield localization
of the projective model structure on the category 
of simplicial presheaves $\mathbf{sPSh}(\dAff) := \Fun(\dAff^\op,\sSet)$
at the set of morphisms $\widehat{W}$ given by the image under the (discrete) 
Yoneda embedding of the weak equivalences $W$ in $\dAff$, see \cite[Section 1.3.1]{ToenVezzosi}.
The model category of \textit{derived stacks}
\begin{flalign}
\dSt\,:=\, \mathcal{L}_{\text{{\'e}t}}\,\dPSt
\end{flalign}
is defined as a further left Bousfield localization at {\'e}tale hypercovers,
see \cite[Section 1.3.2]{ToenVezzosi}.
\end{defi}

\begin{rem} 
The model category of derived affine schemes embeds into the one of derived prestacks 
via a fully faithful model-categorical Yoneda embedding
\begin{flalign}
\begin{gathered}
\xymatrix{
\ar@{-->}[dr] \dAff \ar[rr]^-{\text{Yoneda}}~&~ ~&~\dPSt\\
~&~ \dSt \ar[ru]~&~
}
\end{gathered}
\quad.
\end{flalign}
This embedding preserves weak equivalences as it is constructed from
derived mapping spaces in $\dAff$ and it factorizes through the canonical inclusion 
$\dSt\to \dPSt$ since the {\'e}tale topology is sub-canonical, see \cite[Sections 1.3.1 and 1.3.2]{ToenVezzosi}.
Hence, every derived affine scheme gives rise to a derived stack.
\end{rem}

For the purpose of our work, the most relevant kind of 
derived stacks are derived quotient stacks, in particular those
arising from an action of an affine group 
scheme on a derived affine scheme, see 
Example \ref{ex:quotient} below. 
The action of a (higher) group(oid)
on a derived stack $X_0\in\dSt$ can be encoded in terms of a specific type of 
simplicial object $X_\bullet : \Delta^\op\to \dSt$
in derived stacks (called Segal groupoid object in \cite{ToenVezzosi}),
which one can visualize as
\begin{flalign}
X_\bullet\,=\,\bigg(
\xymatrix{
X_0 \ar[r]~&~\ar@<1ex>[l]\ar@<-1ex>[l] X_1 \ar@<1ex>[r] \ar@<-1ex>[r]~&~\ar@<2ex>[l]\ar@<-2ex>[l]\ar[l] X_2~~\cdots
}\bigg)\quad.
\end{flalign}
The associated derived quotient stack is defined by taking the homotopy colimit 
\begin{flalign}
\vert X_\bullet\vert\,:=\,\hocolim_{\dSt}\big(X_\bullet :\Delta^\op\to \dSt\big)\,\in\,\dSt
\end{flalign}
of this diagram
in the model category of derived stacks from Definition \ref{def:dSt}.
Given any morphism $f_\bullet : X_\bullet\to Y_\bullet$ 
between two simplicial objects $X_\bullet,Y_\bullet : \Delta^\op\to \dSt$,
one obtains from the functoriality of homotopy colimits a morphism
\begin{flalign}\label{eqn:quotientmorphism}
\vert f_\bullet\vert\,:\,\vert X_\bullet\vert~\longrightarrow~\vert Y_\bullet\vert
\end{flalign}
in $\dSt$ between the corresponding derived quotient stacks.
If $f_\bullet$ is a level-wise weak equivalence, i.e.\ $f_n : X_n\to Y_n$
is a weak equivalence in $\dSt$ for all $n\geq 0$, then the induced
morphism \eqref{eqn:quotientmorphism} between the derived quotient stacks
is a weak equivalence in $\dSt$. This is a general consequence of the fact that 
homotopy colimits preserve weak equivalences.
We shall also need the following less direct preservation property for simplicial homotopy 
equivalences of the diagram.
\begin{lem}\label{lem:simplicialhomotopy}
Suppose that the morphism $f_\bullet : X_\bullet\to Y_\bullet$ is quasi-invertible, i.e.\ there 
exists a morphism $g_\bullet : Y_\bullet \to X_\bullet$
and two simplicial homotopies $g_\bullet\,f_\bullet\sim \id_{X_\bullet}$
and $f_\bullet\,g_\bullet\sim \id_{Y_\bullet}$. Then the induced morphism \eqref{eqn:quotientmorphism}
between the derived quotient stacks is a weak equivalence in $\dSt$.
\end{lem}
\begin{proof}
From the Definition \ref{def:dSt} of the model structure on
$\dSt$ in terms of left Bousfield localizations,
we have two left Quillen functors
$\id : \mathbf{sPSh}(\dAff)\to \dPSt$ and $\id: \dPSt\to \dSt$.
Since left Quillen functors preserve homotopy colimits, 
our claim would follow if we can show that the morphism
\begin{flalign}
\hocolim_{\mathbf{sPSh}(\dAff)}\big(X_\bullet) ~\longrightarrow~
\hocolim_{\mathbf{sPSh}(\dAff)}\big(Y_\bullet)
\end{flalign}
between the homotopy colimits with respect to the projective model structure 
is a weak equivalence in $\mathbf{sPSh}(\dAff)$. Since projective weak equivalences
are defined object-wise, the latter amounts to showing that
\begin{flalign}\label{eqn:sSethocolim}
\hocolim_{\sSet}\big(X_\bullet(A)\big)~\longrightarrow~\hocolim_{\sSet}\big(Y_\bullet(A)\big)
\end{flalign}
is a weak equivalence of simplicial sets, for all $A\in\dgCAlg^{\leq 0}$.
The morphism $(f_\bullet)_A : X_\bullet(A) \to Y_\bullet(A)$
between the functors $X_\bullet(A) ,Y_\bullet(A)  : \Delta^\op\to \sSet$ 
can be identified with a morphism 
$(f_{\bullet,\bullet})_A : X_{\bullet,\bullet}(A) \to Y_{\bullet,\bullet}(A)$
between the associated bisimplicial sets $X_{\bullet,\bullet}(A) ,Y_{\bullet,\bullet}(A) : 
\Delta^\op\times\Delta^\op\to \Set$.
The homotopy colimits in \eqref{eqn:sSethocolim} can then 
be computed explicitly by taking diagonals of these bisimplicial sets,
i.e.\ \eqref{eqn:sSethocolim} reduces to the morphism
\begin{flalign}
\mathrm{diag}\big((f_{\bullet,\bullet})_A\big)\,:\,\mathrm{diag}\big(X_{\bullet,\bullet}(A)\big)~\longrightarrow~\mathrm{diag}\big(Y_{\bullet,\bullet}(A)\big)\quad.
\end{flalign}
Our claim then follows from the fact that $\mathrm{diag}$ sends level-wise
weak equivalences of bisimplicial sets, and hence in particular level-wise
simplicial homotopy equivalences, to weak equivalences in $\sSet$, 
see e.g.\ \cite[Chapter IV, Proposition 1.7]{Goerss}.
\end{proof}

\begin{ex}\label{ex:quotient}
The following class of examples will be crucial for our work.
Let $X = \Spec(A)\in\dAff$ be a derived affine scheme
with an action $r : X\times G\to X\,,~(x,g)\mapsto x\,g$ of 
an affine group scheme $G = \Spec(H)\in \mathbf{Grp}(\Aff)$.
One can assemble these data into a simplicial object 
\begin{subequations}\label{eqn:X/Gsimplicial}
\begin{flalign}
N_\bullet(X/G)\,:=\,\bigg(
\xymatrix{
X \ar[r]~&~\ar@<1ex>[l]\ar@<-1ex>[l] X\times G \ar@<1ex>[r] \ar@<-1ex>[r]~&~\ar@<2ex>[l]\ar@<-2ex>[l]\ar[l] X\times G^2~~\cdots
}\bigg)
\end{flalign}
in $\dAff$, and hence via the model-categorical Yoneda embedding in $\dSt$, 
by using the face maps
\begin{flalign}
d_i\,:\, X\times G^n~&\longrightarrow~X\times G^{n-1}\quad,\\
\nn (x,g_1,\dots,g_n)~&\longmapsto~
\begin{cases}
(x\, g_1,g_2,\dots,g_n)~&~~\text{for }i\,=\,0~~,\\
(x,g_1,\dots,g_i\,g_{i+1},\dots,g_n)~&~~\text{for }i\,=\,1,\dots,n-1~~,\\
(x,g_1,\dots,g_{n-1}) ~&~~\text{for }i\,=\,n~~,
\end{cases}
\end{flalign}
and the degeneracy maps
\begin{flalign}
s_i\,:\, X\times G^{n}~\longrightarrow~X\times G^{n+1}~~,
\quad (x,g_1,\dots,g_n)~\longmapsto~(x,g_1,\dots,g_i,e,g_{i+1},\dots, g_n)\quad,
\end{flalign}
\end{subequations}
for all $i=0,\dots,n$, where $e\in G$ denotes the identity element.
We denote by 
\begin{flalign}\label{eqn:X/Ghocolim}
[X/G]\,:=\, \big\vert N_\bullet(X/G)\big\vert \,=\, \hocolim_{\dSt}\big( N_\bullet(X/G): \Delta^\op\to\dSt\big)\,\in\,\dSt
\end{flalign}
the associated derived quotient stack.
\end{ex}

As a final prerequisite for Section \ref{sec:PFA}, we have to recall briefly the concept
of quasi-coherent complexes on derived stacks. Here one encounters potential set 
theoretic size issues, to be addressed in Remark \ref{rem:size} below, 
which for an ease of presentation will be suppressed in the main text.
We denote by $\dgCat$ the $2$-category
of dg-categories, dg-functors and dg-natural transformations over $\bbK$. 
Let us start with the more concrete case of quasi-coherent complexes on derived affine schemes.
For any $\Spec(A)\in \dAff$, we denote by 
\begin{flalign}
\QCoh\big(\Spec(A)\big)\,:=\, {}_A\dgMod_{\mathrm{cof}}\,\in\,\dgCat
\end{flalign} 
the dg-category whose objects
are all (not necessarily bounded) cofibrant $A$-dg-modules $M$
and whose hom-complexes $\hom_A(M,N)$ consist in degree $k\in\bbZ$
of all $A$-linear maps $K :M\to N$ of degree $k$, with differential defined
as usual by $\partial(K) := \dd_N\,K - (-1)^k\,K\,\dd_M $.
For any morphism $f : \Spec(A)\to \Spec(B)$ in $\dAff$, 
we denote by 
\begin{flalign}
\QCoh(f)\,:=\, A\otimes_B(-)\,:\, \QCoh\big(\Spec(B)\big)~\longrightarrow~\QCoh\big(\Spec(A)\big)
\end{flalign} 
the change-of-base dg-functor associated with 
the opposite $\dgCAlg^{\leq 0}$-morphism $f^\ast : B\to A$.
This defines a pseudo-functor
\begin{flalign}\label{eqn:QCohdAff}
\QCoh\,:\, \dAff^\op~\longrightarrow~\dgCat
\end{flalign}
which sends weak equivalences in $\dAff$ to weak equivalences 
in the model structure on dg-categories from \cite{Tabuada}.
(The reason for this is that, for every weak equivalence $f^\ast : B\to A$ in $\dgCAlg^{\leq 0}$ 
and every cofibrant $B$-dg-module $M$, the map $f^\ast\otimes_B \id_M : M \cong B\otimes_B M\to
A\otimes_B M$ is a quasi-isomorphism of $B$-dg-modules.) Quasi-coherent complexes
on derived stacks are then defined by performing a homotopy right Kan extension
of \eqref{eqn:QCohdAff} along the model-categorical Yoneda
embedding $\dAff^\op\to\dSt^\op$. This yields a pseudo-functor (denoted
with abuse of notation by the same symbol)
\begin{flalign}\label{eqn:QCohdSt}
\QCoh\, :\,\dSt^\op~\longrightarrow~\dgCat
\end{flalign}
which sends weak equivalences in $\dSt$ to weak equivalences in $\dgCat$.
The value of this pseudo-functor on a derived quotient stack $[X/G]= \big[\Spec(A)/\Spec(H)\big]\in\dSt$
as in Example \ref{ex:quotient} can be characterized
rather explicitly in terms of $A_\infty$-comodules \cite{holim}. 
In the special case where the affine group scheme $G=\Spec(H)$ is reductive,
one obtains a weakly equivalent but simpler model
\begin{flalign}\label{eqn:QCoh[X/G]}
\QCoh\big([X/G]\big)\,\simeq\, {}_A\dgMod^H_{\mathrm{cof}}\,\in\,\dgCat
\end{flalign}
in terms of cofibrant $A$-dg-modules with a compatible
$H$-coaction, see e.g.\ \cite[Proposition 2.17]{BPSquantization}
for further details on this point. More explicitly,
an object in $\QCoh\big([X/G]\big)$ is a pair $(M,\rho_M)$
consisting of a cofibrant $A$-dg-module $M$
and a coaction $\rho_M: M\to M\otimes H$ which satisfies
$\rho_M(a\,m) = \rho(a)\,\rho_M(m)$, for all $a\in A$ and $m\in M$,
where $\rho:A\to A\otimes H$ denotes the given coaction on $A$.
The hom-complexes $\hom\big((M,\rho_M),(N,\rho_N)\big):= \hom_A(M,N)^H
\subseteq \hom_A(M,N)$ are given by the subcomplexes consisting of 
$A$-linear maps $K:M\to N$ which are strictly $H$-equivariant, i.e.\ 
$(K\otimes\id_H)\,\rho_M = \rho_N\,K$.
\begin{rem}\label{rem:size}
The construction of the pseudo-functors \eqref{eqn:QCohdAff} and 
\eqref{eqn:QCohdSt} is slightly subtle because of size issues. 
Observe that, already in the derived affine case 
\eqref{eqn:QCohdAff}, the dg-categories $\QCoh\big(\Spec(A)\big) = 
{}_A\dgMod_{\mathrm{cof}}$ are not small, hence this assignment
does not take values in Tabuada's model category $\dgCat$ of \textit{small} dg-categories \cite{Tabuada}.
In the general case \eqref{eqn:QCohdSt}, by definition of
$\QCoh(X) = \holim_{\Spec(A)\to X}^{} \QCoh\big(\Spec(A)\big)$
as a homotopy right Kan extension, one encounters homotopy limits 
which are indexed by a slice category $\dAff^\op/X$ that is in general not small, so their existence
is not guaranteed. In the literature on derived algebraic geometry, one can find
the following two solutions to these problems: 1.)~The authors of \cite{GRbook}
work in a more abstract, but powerful setting where ordinary dg-categories are 
replaced by presentable $\bbK$-linear stable $\infty$-categories $\mathbf{DGCat_{\mathrm{cont}}}$,
which via Lurie's general machinery \cite{Lurie} resolves the size issues mentioned above.
2.)~The authors of \cite{ToenVezzosi} provide a more traditional solution
given by using Grothendieck universes in order to introduce bounds on the allowed sizes. 
Relative to a choice of universes $\mathbb{U}\in \mathbb{V}\in\mathbb{W}\cdots$, 
they model derived stacks by functors $\dAff^\op_{\mathbb{U}}\to \sSet_{\mathbb{V}}^{}$.
In this refined context, one obtains a precise definition of \eqref{eqn:QCohdAff} in terms of 
the pseudo-functor $\QCoh: \dAff^\op_{\mathbb{U}}\to \dgCat_{\mathbb{V}}$ taking values in 
$\mathbb{V}$-small dg-categories, which according to \cite{ToenDG} carry
a model structure analogous to the one of \cite{Tabuada}. Since now
the slice categories $\dAff^\op_{\mathbb{U}}/X$ are by construction $\mathbb{V}$-small,
the homotopy right Kan extension \eqref{eqn:QCohdSt} exists. In the present paper,
we prefer to use the second option to resolve size issues as it allows us to remain
in the context of model categories. To avoid cluttering notations, we will continue to 
drop in the main text all labels referring to Grothendieck universes.
\end{rem}

\subsection{\label{subsec:dCrit}Derived critical loci}
The concept of a \textit{derived critical locus} is a derived 
geometric refinement of the set of critical points of a function.
In the context of mathematical physics, this function is usually 
the action function of a physical system,
so that the derived critical locus describes a derived geometric model
for the moduli space of solutions to the associated Euler-Lagrange equations.
\sk

Given a $\dSt$-morphism $S : Y\to \mathbb{A}^1$ 
from a derived Artin stack $Y$ to the affine line $\mathbb{A}^1 := \Spec\big(\bbK[x]\big)$,
the derived critical locus is defined as the homotopy pullback
\begin{flalign}
\begin{gathered}
\xymatrix{
\dCrit(S) \ar@{-->}[r]\ar@{-->}[d]~&~Y \ar[d]^-{\dd_{\dR}S}\\
Y \ar[r]_-{0}~&~T^\ast Y
}
\end{gathered}
\end{flalign}
in the model category $\dSt$, where $0$ denotes the zero-section of the cotangent bundle $T^\ast Y$
and $\dd_{\dR}S$ denotes the section obtained by applying the de Rham differential to $S$.
Since all homotopy pullbacks exist in $\dSt$, the derived critical locus $\dCrit(S)\in \dSt$
always exists for any $S$, however its explicit description is in general difficult.
Concrete models for derived critical loci have been developed for the following special 
cases: 1.)~$Y=\Spec(A)$ is an ordinary affine scheme \cite{Vezzosi},
2.)~$Y=[\Spec(A)/\g]$ is a formal quotient stack \cite{CG2}, and 3.)~$Y=[\Spec(A)/\Spec(H)]$
is the quotient stack associated with the action of an affine group scheme on an ordinary (smooth) 
affine scheme \cite{BSSdCrit,AnelCalaque}. See also \cite{Albin} for a generalization to
Lie algebroids and groupoids.
\sk

Since it will be needed in the main text, we
briefly recall the explicit model from \cite{BSSdCrit} 
for the derived critical locus of a function 
\begin{flalign}
S \,:\, [X/G]\,=\,[\Spec(A)/\Spec(H)]~\longrightarrow~ \mathbb{A}^1
\end{flalign} 
on the quotient stack associated with the action $r: X\times G\to X$ of 
an affine group scheme $G=\Spec(H)$ on an ordinary (smooth) affine scheme $X=\Spec(A)$.
In this case the derived critical locus is a derived quotient stack
\begin{flalign}\label{eqn:dCritquotient}
\dCrit(S)\,\simeq\,[Z/G]\,\in\,\dSt
\end{flalign}
of a derived affine scheme $Z = \Spec\big(\O(Z)\big)\in\dAff$ by an action of $G$. The commutative
dg-algebra $\O(Z)\in\dgCAlg^{\leq 0}$ specifying $Z$ is given by the graded commutative algebra
\begin{subequations}\label{eqn:OZ}
\begin{flalign}\label{eqn:OZ1}
\O(Z)\,=\,\Sym_A\Big( \big(A\otimes\g[2]\big)\oplus \mathrm{T}_{\! A}[1]\Big)
\end{flalign}
which is generated over $A$ 
by the free $A$-module $A\otimes\g[2]$, where $\g$ denotes the
Lie algebra of $G$, and the $[1]$-shift of the $A$-module $\mathrm{T}_{\! A}$ 
of derivations of $A$. The differential of $\O(Z)$ is defined on 
the generators by
\begin{flalign}\label{eqn:OZ2}
\dd a\,=\, 0\quad,\qquad
\dd v\,=\, \iota_v \dd_{\dR} S \quad ,\qquad
\dd\xi \,=\,-\iota_{\rho(\xi)}\lambda \quad,
\end{flalign}
\end{subequations}
for all $a\in A$, $v\in \mathrm{T}_{\! A}[1]$ and $\xi\in\g[2]$.
The second expression denotes the contraction between the derivation
$v\in \mathrm{T}_{\! A} $ and the $1$-form $\dd_{\dR} S\in\Omega^1_A$.
In the third expression, 
$\lambda\in \Omega^1_{\Sym_A\mathrm{T}_{\! A}}$ denotes the tautological $1$-form on $T^\ast X$
and $\rho : \g\to \mathrm{T}_{\Sym_A\mathrm{T}_{\! A}}$
denotes the Lie algebra action which is induced from the $G$-action
on the cotangent bundle $T^\ast X = \Spec(\Sym_A\mathrm{T}_{\! A})$.
The $G$-action $r:Z\times G\to Z$ entering \eqref{eqn:dCritquotient}
is induced from the given $G$-actions
on $X$ and $T^\ast X$ and the adjoint action on the Lie algebra $\g$.

\subsection{\label{subsec:latticeYM}Lattice gauge theory}
We recall some basic aspects of lattice gauge theory on 
the $2$-dimensional square lattice $\bbZ^2$.
We denote points by $x = (x_1,x_2)\in\bbZ^2$ and interpret $x_1,x_2\in\bbZ$ as discrete coordinates.
\sk

To describe a gauge theory on $\bbZ^2$, one has to choose a structure group, which
for simplicity we shall always take to be the general linear group 
$\GL_n$ of some finite degree $n\in\bbN$ over the field $\bbK$.
A gauge field (or connection) on $\bbZ^2$ is given by an assignment of structure group elements
$T_i(x)\in\GL_n$ to the edges of $\bbZ^2$, i.e.\
\begin{flalign}\label{eqn:latticepicture}
\begin{gathered}
\begin{tikzpicture}[scale=1.1,decoration={
    markings,
    mark=at position 0.5 with {\arrow{>}}}
    ]
\draw[color=blue,postaction={decorate},thick] (0,-0.25) -- (0,2.25)  node [midway, left] {{\footnotesize $T_2(x_1,x_2)$}}  ;
\draw[color=blue,postaction={decorate},thick] (2,-0.25)  -- (2,2.25) node [midway, right] {{\footnotesize $T_2(x_1+1,x_2)$}} ;
\draw[color=blue,postaction={decorate},thick] (-0.25,0)  -- (2.25,0) node [midway, below] {{\footnotesize $T_1(x_1,x_2)$}} ;
\draw[color=blue,postaction={decorate},thick] (-0.25,2)  -- (2.25,2) node [midway, above] {{\footnotesize $T_1(x_1,x_2+1)$}} ;
\filldraw (0,0) circle (1.5pt);
\filldraw (0,2) circle (1.5pt);
\filldraw (2,0) circle (1.5pt);
\filldraw (2,2) circle (1.5pt);
\node at (0,-0.5) {{\footnotesize $x_1$}};
\node at (2,-0.5) {{\footnotesize $x_1+1$}};
\node at (-0.75,0) {{\footnotesize $x_2$}};
\node at (-0.75,2) {{\footnotesize $x_2+1$}};
\end{tikzpicture}
\end{gathered}\qquad.
\end{flalign}
One interprets these group elements as parallel transports along the edges.
The space of connections on $\bbZ^2$ is thus given by the product
\begin{flalign}\label{eqn:Con}
\Con(\bbZ^2)\,:=\,\prod_{(x,i)\in\bbZ^2\times\{1,2\}} \GL_n\quad,
\end{flalign}
where as visualized in \eqref{eqn:latticepicture} we use the index 
$i=1$ for the $x_1$-components $T_1(x)$ of the connection 
and $i=2$ for the $x_2$-components $T_2(x)$. A gauge transformation
in this discrete context is given by an assignment of structure group elements
$U(x)\in\GL_n$ to the vertices of $\bbZ^2$, i.e.\ the black dots in \eqref{eqn:latticepicture}.
The gauge group on $\bbZ^2$ is thus given by the product group
\begin{flalign}\label{eqn:Gau}
\G(\bbZ^2)\,:=\,\prod_{x \in\bbZ^2} \GL_n\quad.
\end{flalign}
The action of gauge transformations on connections is given by
\begin{flalign}\label{eqn:Gauaction}
r\,:\,\Con(\bbZ^2)\times \G(\bbZ^2)~&\longrightarrow~\Con(\bbZ^2)\quad,\\
\nn \Big(\big(T_i(x)\big)_{(x,i)}, \big(U(x)\big)_{x}\Big)~&\longmapsto
\Big( U(x+e_i)^{-1}\,T_{i}(x)\,U(x)\Big)_{(x,i)}\quad,
\end{flalign}
where $e_i\in\bbZ^2$ is defined by $e_1=(1,0)$ and $e_2=(0,1)$.
This means that $T_{i}(x)$ transforms by right multiplication
with the group element $U(x)$ located at the source of the edge $(x,i)\in\bbZ^2\times\{1,2\}$
and by left multiplication with the inverse of the group element 
$U(x+e_i)$ located at the target.
\sk

It remains to specify a gauge invariant action function to encode the dynamics
of our lattice gauge theory. For this we shall take the Wilson action \cite{Wilson},
which is a discrete approximation of the Yang-Mills action from continuum gauge theory.
The basic idea is to consider the, say counter-clockwise,
parallel transports along the $2$-dimensional faces in \eqref{eqn:latticepicture},
which we denote by
\begin{flalign}\label{eqn:Efield}
E(x)\,:=\, T_2(x)^{-1}\,T_1(x + e_2 )^{-1}\,T_2(x + e_1)\, T_1(x)\,\in\,\GL_n\quad,
\end{flalign}
for all $x\in\bbZ^2$. The Wilson action is encoded by the family of
gauge invariant functions
\begin{flalign}\label{eqn:Wilsonaction}
S_\Lambda\,:\,\Con(\bbZ^2)~\longrightarrow~\mathbb{A}^1~~,\quad
\big(T_{i}(x)\big)_{(x,i)}~\longmapsto~\sum_{x\in\Lambda\subset \bbZ^2}
\mathrm{Tr}\big(E(x)\big)\quad,
\end{flalign}
which is labeled by all finite subsets $\Lambda\subset \bbZ^2$ whose role is to make
the summation over $x$ well defined. The need for such regulators $\Lambda$
for the action is typical for field theories on non-compact spaces.
The standard method to derive Euler-Lagrange equations from the family of actions
$\{S_\Lambda\}$ is as follows: Given any compactly supported
variation $\delta_\alpha$, one chooses a sufficiently large $\Lambda\subset \bbZ^2$
such that the support $\supp(\alpha) \ll \Lambda$ is safely contained (to avoid cutoff effects)
and then one computes as usual the Euler-Lagrange equations from $\delta_\alpha S_\Lambda =0$.
Applying this procedure to \eqref{eqn:Wilsonaction} yields the Euler-Lagrange equations
\begin{subequations}\label{eqn:EL}
\begin{flalign}
\label{eqn:EL1}E(x)\,=\, T_{1}(x-e_1)\,E(x-e_1)
\,T_1(x-e_1)^{-1}\quad,\\
\label{eqn:EL2}E(x)\,=\, T_{2}(x-e_2)\,E(x-e_2)
\,T_2(x-e_2)^{-1}\quad,
\end{flalign}
\end{subequations}
for all $x\in\bbZ^2$.


\section{\label{sec:dCrit}Global derived critical locus of $2$d lattice Yang-Mills theory}
In this section we describe explicitly the derived critical locus from Subsection
\ref{subsec:dCrit} for the lattice Yang-Mills model from Subsection \ref{subsec:latticeYM}.
There are some minor subtleties in working out this description, arising from the fact
that our discrete spacetime $\bbZ^2$ is non-compact,
which however can be controlled via standard methods from field theory, such as the regularized
actions from Subsection \ref{subsec:latticeYM}.

\subsection{\label{subsec:GLn}Some computational aspects of $\GL_n$ and $\gl_n$}
Before we start, let us recall some basic aspects of
the affine group scheme $\GL_n$ which are essential for our discussion below.
The associated commutative algebra $\O(\GL_n)\in\CAlg$ of $\GL_n = \Spec\big(\O(\GL_n)\big)$
is given by the localization
\begin{flalign}
\O(\GL_n)\,:=\, \bbK[\{T_{ab}\}][(\det T)^{-1}]\,:=\,
\bbK[\{T_{ab}\},\widetilde{T}]\big/\big((\det T)\,\widetilde{T}-1\big)
\end{flalign}
of the polynomial algebra with $n^2$ generators $T_{ab}$, for $a,b=1,\dots,n$, at the determinant
of the $n\times n$-matrix $T = (T_{ab})$ which is formed by the generators.
The group structure of $\GL_n$ is encoded by the following Hopf algebra structure
on $\O(\GL_n)$
\begin{subequations}
\begin{flalign}\label{eqn:Hopfcomponents}
\Delta(T_{ab}) \,=\, \sum_{c=1}^n T_{ac}\otimes T_{cb}\quad , \qquad
\epsilon(T_{ab})\,=\,\delta_{ab}\quad,\qquad
S(T_{ab}) \,=\,T^{-1}_{ab}\quad, 
\end{flalign}
where $\delta_{ab}$ denotes the Kronecker delta and $T^{-1}_{ab}$ are the entries of the inverse $T^{-1}$ of the 
matrix $T$, which exists since we have localized at the determinant
$\det T$. This Hopf algebra structure can be written more conveniently in matrix notation as
\begin{flalign}
\Delta(T)\,=\,T\otimes T\quad,\qquad \epsilon(T)\,=\,\oone\quad,\qquad
S(T)\,=\,T^{-1}\quad,
\end{flalign}
\end{subequations}
where $T\otimes T$ denotes the combination of matrix multiplication 
and tensor product from \eqref{eqn:Hopfcomponents} and $\oone$ denotes the identity matrix.
\sk

The Lie algebra of $\GL_n$ can be defined in terms
of derivations $\gl_n :=\mathrm{Der}_\epsilon\big(\O(\GL_n),\bbK\big)$
relative to the counit $\epsilon : \O(\GL_n)\to\bbK$, i.e.\
linear maps $\xi : \O(\GL_n)\to\bbK$ which satisfy $\xi(h\,k)= \xi(h)\,\epsilon(k) + \epsilon(h)\,\xi(k)$,
for all $h,k\in\O(\GL_n)$. A basis for $\gl_n$ is given by the derivations $\xi_{ab}\in \gl_n$,
for $a,b=1,\dots,n$, which are defined on the generators $T_{cd}$ of $\O(\GL_n)$ by
\begin{flalign}
\xi_{ab}(T_{cd})\,=\, \delta_{ad}\,\delta_{bc}\quad.
\end{flalign}
One can think of these derivations in terms of partial derivatives 
$\xi_{ab} = \frac{\partial~~}{\partial T_{ba}}$. It will often be
convenient to assemble these basis derivations into an $n\times n$-matrix $\xi = (\xi_{ab})$.
\sk

The right adjoint action $\Ad:\GL_n\times\GL_n\to \GL_n\,,~(g^\prime,g)\mapsto g^{-1}\,g^\prime\,g$
is encoded algebraically by the right adjoint coaction which in matrix notation is given by
\begin{subequations}
\begin{flalign}
\rho \,:\, \O(\GL_n)~\longrightarrow~ \O(\GL_n)\otimes \O(\GL_n)~~,\quad
T~\longmapsto~U^{-1}\,T\,U\quad,
\end{flalign}
where we have identified
\begin{flalign}
\O(\GL_n)\otimes \O(\GL_n)\,\cong\, \bbK [\{T_{ab}\},\{U_{ab}\}][(\det T)^{-1}, (\det U)^{-1}]\quad.
\end{flalign}
\end{subequations}
The right adjoint coaction on the Lie algebra $\gl_n$ is given in matrix notation by
\begin{flalign}
\rho \,:\, \gl_n~\longrightarrow~ \gl_n \otimes \O(\GL_n)~~,\quad
\xi~\longmapsto~ U^{-1}\,\xi\,U\quad.
\end{flalign}
Furthermore, the Lie algebra $\gl_n$ acts via left invariant derivations on $\O(\GL_n)$,
i.e.\ there is a Lie algebra representation
\begin{subequations}
\begin{flalign}
\rho^L\,:\, \gl_n~\longrightarrow~\mathrm{T}_{\!\O(\GL_n)}
\end{flalign}
which on generators reads as
\begin{flalign}
\rho^L(\xi_{ab})(T_{cd})\,:=\,\big(\id_{\O(\GL_n)}\otimes \xi_{ab}\big)\Delta(T_{cd})\,=\, T_{cb}\,\delta_{ad}\quad.
\end{flalign}
\end{subequations}
The linear map $\rho^L$ provides a trivialization
\begin{flalign}\label{eqn:TastGtrivialization}
\O(\GL_n)\otimes\gl_n ~\stackrel{\cong}{\longrightarrow}~\mathrm{T}_{\!\O(\GL_n)}
\end{flalign}
of the $\O(\GL_n)$-module of derivations $\mathrm{T}_{\!\O(\GL_n)}$.

\subsection{\label{subsec:dCritYM}Description of $\mathrm{dCrit}(S)$}
We now describe the derived critical locus \eqref{eqn:dCritquotient}
for the example given by the lattice Yang-Mills model from Subsection \ref{subsec:latticeYM}.
In this example, the affine scheme $X=\Spec(A)$ is given by
the space of connections \eqref{eqn:Con}, i.e.\ we have that
\begin{flalign}\label{eqn:OCon}
A \,=\,\O\big(\Con(\bbZ^2)\big) \,=\, \bigotimes_{(x,i)\in\bbZ^2\times \{1,2\}}\O(\GL_n)
\end{flalign}
is an (infinite) coproduct in $\CAlg$. We denote by $T_i(x)\in \O\big(\Con(\bbZ^2)\big)$ the generators
of this commutative algebra which are given by $T\in \O(\GL_n)$ on the tensor
factor $(x,i)$ and $1\in\O(\GL_n)$ on all other tensor factors.
The affine group scheme $G = \Spec(H)$ is given by the gauge group \eqref{eqn:Gau}, i.e.\
we have that
\begin{flalign}\label{eqn:OGau}
H\,=\,\O\big(\G(\bbZ^2)\big)\,=\, \bigotimes_{x\in\bbZ^2}\O(\GL_n)
\end{flalign}
is an (infinite) coproduct of commutative Hopf algebras. We denote by $U(x)\in \O\big(\G(\bbZ^2)\big)$
the generators of this commutative Hopf algebra which are given by $U\in \O(\GL_n)\cong
\bbK[\{U_{ab}\}][(\det U)^{-1}]$ 
on the tensor factor $x$ and $1\in\O(\GL_n)$ on all other tensor factors.
The action \eqref{eqn:Gauaction} of gauge transformations on connections
is given algebraically by the coaction
\begin{flalign}\label{eqn:gaugeaction}
\rho \,:\, \O\big(\Con(\bbZ^2)\big)~&\longrightarrow~\O\big(\Con(\bbZ^2)\big)\otimes \O\big(\G(\bbZ^2)\big)\quad,\\
\nn T_{i}(x)~&\longmapsto~ U(x+e_i)^{-1} \,T_{i}(x)\,U(x)\quad.
\end{flalign}

To determine the derived affine scheme $Z(\bbZ^2)=\Spec\big(\O\big(Z(\bbZ^2)\big)\big)\in\dAff$ 
from \eqref{eqn:OZ} which enters
the derived critical locus $\dCrit(S)\simeq \big[Z(\bbZ^2)/\G(\bbZ^2)\big]$, we have to describe the Lie algebra
$\g(\bbZ^2)$ of the gauge group $\G(\bbZ^2)$ and the module of derivations
$\mathrm{T}_{\! \O(\Con(\bbZ^2))}$ on the space of connections $\Con(\bbZ^2)$.
Since $\O\big(\G(\bbZ^2)\big)$ and $\O\big(\Con(\bbZ^2)\big)$ are infinitely generated
as a consequence of the non-compactness of the discrete spacetime, 
there exist different concepts of derivations which are distinguished by
their support properties on $\bbZ^2$. Observing that
\eqref{eqn:OCon} and \eqref{eqn:OGau} describe functions
which are compactly supported on $\bbZ^2$ due to the coproducts,
we will model $\g(\bbZ^2)$ and $\mathrm{T}_{\! \O(\Con(\bbZ^2))}$ by 
derivations which are compactly supported on $\bbZ^2$ too.
This yields
\begin{flalign}
\g(\bbZ^2)\,=\, \bigoplus_{x\in\bbZ^2} \gl_n
\end{flalign}
and, recalling also the trivialization \eqref{eqn:TastGtrivialization},
\begin{flalign}\label{eqn:cotangentCon}
\mathrm{T}_{\! \O(\Con(\bbZ^2))}\,= \,\O\big(\Con(\bbZ^2)\big)\otimes\bigoplus_{(x,i)\in\bbZ^2\times\{1,2\}}\gl_n\quad.
\end{flalign}
We denote by $\xi(x)\in \g(\bbZ^2)$
the element which is given by $\xi\in \gl_n$ 
on the summand $x$ and $0\in\gl_n$ on all other summands.
Similarly, we denote by $\xi_{i}(x)\in \mathrm{T}_{\! \O(\Con(\bbZ^2))}$ 
the element which is given by $\xi\in \gl_n$ 
on the summand $(x,i)$ and $0\in\gl_n$ on all other summands.
\sk

Combining the above building blocks, we obtain that the graded commutative algebra
\eqref{eqn:OZ1} reads in our example as
\begin{subequations}\label{eqn:OZYM}
\begin{flalign}\label{eqn:OZYM1}
\O\big(Z(\bbZ^2)\big)\,\cong\,\bigotimes_{x\in\bbZ^2}\Sym \big(\gl_n[2]\big)\otimes \!\! \bigotimes_{(x,i)\in\bbZ^2\times\{1,2\}} \!\! \Sym\big(\gl_n[1]\big)\otimes \!\! \bigotimes_{(x,i)\in\bbZ^2\times\{1,2\}} \!\! \O(\GL_n)\quad.
\end{flalign}
To determine the differential \eqref{eqn:OZ2} for our example, 
let us note that both the action $S_\Lambda$ \eqref{eqn:Wilsonaction} and 
the tautological $1$-form $\lambda_{\Lambda}$ on $T^\ast\Con(\bbZ^2)$
require a regulator $\Lambda\subset \bbZ^2$ to be well defined.
This regulator can however be easily removed $\Lambda\to \bbZ^2$ in the 
differential because  both contractions in \eqref{eqn:OZ2} 
are against derivations which are compactly supported on $\bbZ^2$.
One then finds the following explicit expressions
for the differential on the generators of $\O\big(Z(\bbZ^2)\big)$
\begin{flalign}\label{eqn:OZYM2}
\dd T_i(x)\,&=\,0\quad,\\[3pt]
\dd \xi_1(x)\,&=\, E(x)- T_{2}(x-e_2)\,E(x-e_2)\,T_2(x-e_2)^{-1}\quad,\\[3pt]
\dd \xi_2(x)\,&=\, T_{1}(x-e_1)\,E(x-e_1)\,T_1(x-e_1)^{-1}-E(x)\quad,\\[3pt]
\nn \dd\xi(x)\,&=\,- \xi_1(x) + T_{1}(x-e_1)\,\xi_1(x-e_1)\,T_1(x-e_1)^{-1}  \\
&\qquad 
-\xi_2(x)  + T_{2}(x-e_2)\,\xi_2(x-e_2)\,T_2(x-e_2)^{-1}\quad,
\end{flalign}
\end{subequations}
where we recall that $E(x)$ has been defined in \eqref{eqn:Efield}.
Note that the differential on the degree $-1$ generators $\xi_i(x)$
is given by the Euler-Lagrange equations \eqref{eqn:EL}. The differential
on the degree $-2$ generators $\xi(x)$ has been determined 
also in \cite[Section 4.1]{BPSquantization} for the case where $\bbZ^2$ is replaced
by a finite directed graph.
\sk

To conclude the description of the derived critical locus $\dCrit(S)\simeq \big[Z(\bbZ^2)/\G(\bbZ^2)\big]$, 
we note that the action $r:Z(\bbZ^2)\times \G(\bbZ^2)\to Z(\bbZ^2)$
of the gauge group is given algebraically by the coaction
$\rho : \O\big(Z(\bbZ^2)\big)\to \O\big(Z(\bbZ^2)\big)\otimes \O\big(\G(\bbZ^2)\big)$ which reads on
the generators of $\O\big(Z(\bbZ^2)\big)$ as
\begin{subequations}\label{eqn:rhodCrit}
\begin{flalign}
\rho\big(T_i(x)\big)\,&=\,U(x+e_i)^{-1}\,T_{i}(x)\,U(x)\quad,\\[3pt]
\rho\big(\xi_i(x)\big)\,&=\,U(x)^{-1}\,\xi_{i}(x)\,U(x)\quad,\\[3pt]
\rho\big(\xi(x)\big)\,&=\,U(x)^{-1}\,\xi(x)\,U(x)\quad.
\end{flalign}
\end{subequations}

\subsection{\label{subsec:axialgauge}Axial gauge fixing}
In preparation for the proof of Theorem \ref{theo:locallyconstant} below,
we provide a weakly equivalent description of the
derived critical locus $\dCrit(S)\simeq \big[Z(\bbZ^2)/\G(\bbZ^2)\big]$
from Subsection \ref{subsec:dCritYM} which implements an axial gauge fixing. 
Note that there exist two different axial gauge fixings on the $2$-dimensional square lattice 
$\bbZ^2$, enforcing either that $T_1(x)=\oone$, for all $x\in\bbZ^2$, or that
$T_2(x)=\oone$, for all $x\in\bbZ^2$. Since the two components of the connection
enter symmetrically (up to signs) in \eqref{eqn:OZYM}, it suffices to discuss 
only one of these axial gauge fixings, say $T_2(x)=\oone$ for all $x\in\bbZ^2$. The other 
one will then follow by making some evident minor adaptations to the construction below.
\sk

Let us now formalize this gauge fixing procedure. We define the
affine scheme of connections in axial gauge by
\begin{flalign}\label{eqn:Congf}
\Con^{\gf}(\bbZ^2)\,:=\,\prod_{x\in\bbZ^2}\GL_n\quad,\qquad
\O\big(\Con^{\gf}(\bbZ^2)\big)\,=\,\bigotimes_{x\in\bbZ^2}\O(\GL_n)
\end{flalign}
and consider the embedding $j : \Con^{\gf}(\bbZ^2)\to \Con(\bbZ^2)$
into the affine scheme of connections \eqref{eqn:Con}
which is defined by the $\CAlg$-morphism
\begin{flalign}\label{eqn:jmap}
j^\ast\,:\,\O\big(\Con(\bbZ^2)\big)~&\longrightarrow~\O\big(\Con^{\gf}(\bbZ^2)\big)\quad,\\
\nn T_{1}(x)~&\longmapsto~T(x)\quad,\\
\nn T_2(x)~&\longmapsto~\oone\quad,
\end{flalign}
where $T(x)\in\O\big(\Con^{\gf}(\bbZ^2)\big)$ denote the generators of \eqref{eqn:Congf}.
We also define the affine group scheme of gauge transformations in axial gauge by
\begin{flalign}\label{eqn:Gaugf}
\G^{\gf}(\bbZ^2)\,:=\,\prod_{x_1\in\bbZ}\GL_n\quad,\qquad
\O\big(\G^{\gf}(\bbZ^2)\big)\,=\,\bigotimes_{x_1\in\bbZ}\O(\GL_n)
\end{flalign}
and consider the embedding $\und{j} : \G^{\gf}(\bbZ^2)\to \G(\bbZ^2)$
into the affine group scheme of gauge transformations \eqref{eqn:Gau}
which is defined by the commutative Hopf algebra morphism
\begin{flalign}\label{eqn:junderlinemap}
\und{j}^\ast\,:\,\O\big(\G(\bbZ^2)\big)~\longrightarrow~\O\big(\G^{\gf}(\bbZ^2)\big)~~,\quad
U(x)~&\longmapsto~U(x_1)\quad,
\end{flalign}
where $U(x_1)\in\O\big(\G^{\gf}(\bbZ^2)\big)$ denote the generators of 
\eqref{eqn:Gaugf}.
Let us further consider the action  $r:\Con^{\gf}(\bbZ^2)\times \G^{\gf}(\bbZ^2) \to \Con^{\gf}(\bbZ^2)$ which
is defined by the coaction
\begin{flalign}
\rho \,:\, \O\big(\Con^{\gf}(\bbZ^2)\big)~&\longrightarrow~\O\big(\Con^{\gf}(\bbZ^2)\big)\otimes \O\big(\G^{\gf}(\bbZ^2)\big)\quad,\\
\nn T(x)~&\longmapsto~ U(x_1+1)^{-1}\,T(x)\, U(x_1)\quad.
\end{flalign}
One then directly checks that the diagram
\begin{flalign}
\begin{gathered}
\xymatrix{
\ar[d]_-{j\times \und{j}}\Con^{\gf}(\bbZ^2)\times \G^{\gf}(\bbZ^2)\ar[r]^-{r} ~&~\Con^{\gf}(\bbZ^2)\ar[d]^-{j}\\
\Con(\bbZ^2)\times \G(\bbZ^2) \ar[r]_-{r}~&~\Con(\bbZ^2)
}
\end{gathered}
\end{flalign}
in $\Aff$ commutes, where the bottom horizontal arrow is the action \eqref{eqn:gaugeaction}.
This means that the embedding $j : \Con^{\gf}(\bbZ^2)\to \Con(\bbZ^2)$
is equivariant relative to the affine group scheme morphism $\und{j} : \G^{\gf}(\bbZ^2)\to \G(\bbZ^2)$.
\sk

To implement the axial gauge fixing in the derived critical locus
$\dCrit(S)\simeq \big[Z(\bbZ^2)/\G(\bbZ^2)\big]$, we observe that 
$\O\big(Z(\bbZ^2)\big)\in\dgCAlg^{\leq 0}$ in \eqref{eqn:OZYM}
is a commutative dg-algebra over  $\O\big(\Con(\bbZ^2)\big)\in\CAlg$.
Performing a change-of-base along the $\CAlg$-morphism \eqref{eqn:jmap},
we define
\begin{subequations}\label{eqn:OZYMtg}
\begin{flalign}\label{eqn:OZYMtg1}
\nn \O\big(Z^{\gf}(\bbZ^2)\big)\,&:=\, \O\big(\Con^{\gf}(\bbZ^2)\big) 
\otimes_{\O(\Con(\bbZ^2))}^{} \O\big(Z(\bbZ^2)\big)\\[3pt]
\,&\,\cong\,\bigotimes_{x\in\bbZ^2}\Sym \big(\gl_n[2]\big)\otimes 
\!\!\bigotimes_{(x,i)\in\bbZ^2\times\{1,2\}}\!\!\Sym\big(\gl_n[1]\big)\otimes
\bigotimes_{x\in\bbZ^2}\O(\GL_n)\quad.
\end{flalign}
The induced differential on this commutative dg-algebra is given explicitly by
\begin{flalign}\label{eqn:OZYMtg2}
\dd T(x)\,&=\,0\quad,\\[3pt]
\dd \xi_1(x)\,&=\, E^{\gf}(x)- E^{\gf}(x-e_2)\quad,\\[3pt]
\dd \xi_2(x)\,&=\, T(x-e_1)\,E^{\gf}(x-e_1)\,T(x-e_1)^{-1} -E^{\gf}(x)\quad,\\[3pt]
\dd\xi(x)\,&=\, -\xi_1(x) + T(x-e_1)\,\xi_1(x-e_1)\,T(x-e_1)^{-1} 
-\xi_2(x)  + \xi_2(x-e_2)\quad,
\end{flalign}
\end{subequations}
where $E^{\gf}(x)$ is defined by inserting $T_2(x)=\oone$ and 
$T_1(x)=T(x)$ into \eqref{eqn:Efield}, i.e.\
\begin{flalign}\label{eqn:Efieldgf}
E^{\gf}(x)\,:=\,T(x+e_2)^{-1}\, T(x)\quad.
\end{flalign}
Observe that there exists an induced action
$r: Z^{\gf}(\bbZ^2)\times \G^{\gf}(\bbZ^2)\to Z^{\gf}(\bbZ^2)$ 
of the gauge transformations in axial gauge \eqref{eqn:Gaugf}
which is given algebraically by the coaction
$\rho : \O\big(Z^{\gf}(\bbZ^2)\big)\to \O\big(Z^{\gf}(\bbZ^2)\big)\otimes \O\big(\G^{\gf}(\bbZ^2)\big)$ 
that is defined on the generators of $\O\big(Z^{\gf}(\bbZ^2)\big)$ by
\begin{subequations}\label{eqn:rhodCrittg}
\begin{flalign}
\rho\big(T(x)\big)\,&=\,U(x_1+1)^{-1}\,T(x)\,U(x_1)\quad,\\[3pt]
\rho\big(\xi_i(x)\big)\,&=\,U(x_1)^{-1}\,\xi_{i}(x)\,U(x_1)\quad,\\[3pt]
\rho\big(\xi(x)\big)\,&=\,U(x_1)^{-1}\,\xi(x)\,U(x_1)\quad.
\end{flalign}
\end{subequations}
Let us further observe that \eqref{eqn:jmap} induces a morphism
$j : Z^{\gf}(\bbZ^2)\to Z(\bbZ^2)$ of derived affine schemes
whose opposite $\dgCAlg^{\leq 0}$-morphism reads explicitly as
\begin{flalign}\label{eqn:jmapZ}
 j^\ast\,:\,\O\big(Z(\bbZ^2)\big)~&\longrightarrow~\O\big(Z^{\gf}(\bbZ^2)\big)\quad,\\
\nn T_{1}(x)~&\longmapsto~T(x)\quad,\\
\nn T_2(x)~&\longmapsto~\oone\quad,\\
\nn \xi_i(x)~&\longmapsto~\xi_i(x)\quad,\\
\nn \xi(x)~&\longmapsto~\xi(x)\quad.
\end{flalign}
This morphism is equivariant relative to the affine group scheme morphism 
$\und{j} : \G^{\gf}(\bbZ^2)\to \G(\bbZ^2)$ from \eqref{eqn:junderlinemap}, i.e.\
the diagram
\begin{flalign}
\begin{gathered}
\xymatrix{
\ar[d]_-{j\times \und{j}}Z^{\gf}(\bbZ^2)\times \G^{\gf}(\bbZ^2)\ar[r]^-{r} ~&~Z^{\gf}(\bbZ^2)\ar[d]^-{j}\\
Z(\bbZ^2)\times \G(\bbZ^2) \ar[r]_-{r}~&~Z(\bbZ^2)
}
\end{gathered}
\end{flalign}
in $\dAff$ commutes. This allows us to define a morphism
\begin{flalign}\label{eqn:Jmapsimplicial}
\begin{gathered}
\xymatrix{
\ar[d]_-{j}Z^{\gf}(\bbZ^2) \ar[r]~&~\ar@<1ex>[l]\ar@<-1ex>[l] \ar[d]_-{j\times\und{j}} 
Z^{\gf}(\bbZ^2)\times \G^{\gf}(\bbZ^2)
\ar@<1ex>[r] \ar@<-1ex>[r]~&~\ar@<2ex>[l]\ar@<-2ex>[l]\ar[l] \ar[d]_-{j\times\und{j}^2} 
Z^{\gf}(\bbZ^2)\times \G^{\gf}(\bbZ^2)^{2}~~\cdots\\
Z(\bbZ^2) \ar[r]~&~\ar@<1ex>[l]\ar@<-1ex>[l] 
Z(\bbZ^2)\times \G(\bbZ^2) 
\ar@<1ex>[r] \ar@<-1ex>[r]~&~\ar@<2ex>[l]\ar@<-2ex>[l]\ar[l] 
Z(\bbZ^2)\times \G(\bbZ^2)^{2}~~\cdots
}
\end{gathered}
\end{flalign}
of simplicial diagrams in $\dAff$ as in \eqref{eqn:X/Gsimplicial}, and hence by
passing to the homotopy colimits \eqref{eqn:X/Ghocolim} a $\dSt$-morphism
\begin{flalign}\label{eqn:Jmap}
J\,:\,\dCrit^{\gf}(S)\,:=\,\big[Z^{\gf}(\bbZ^2)/\G^{\gf}(\bbZ^2)\big]~\longrightarrow~\big[Z(\bbZ^2)/\G(\bbZ^2)\big]\,\simeq\,\dCrit(S)
\end{flalign}
between the associated derived quotient stacks. 
\sk

We would like to prove that \eqref{eqn:Jmap} is a weak equivalence
in the model category $\dSt$ of derived stacks, which then provides
our desired weakly equivalent model $\dCrit^{\gf}(S)=\big[Z^{\gf}(\bbZ^2)/\G^{\gf}(\bbZ^2)\big]$
for the derived critical locus $\dCrit(S)\simeq \big[Z(\bbZ^2)/\G(\bbZ^2)\big]$.
Our strategy is to apply Lemma \ref{lem:simplicialhomotopy}, i.e.\ 
we have to find a quasi-inverse for the morphism \eqref{eqn:Jmapsimplicial} of simplicial diagrams.
The key ingredient entering our construction below is given by the elements
$\widehat{T}(x)\in\O\big(Z(\bbZ^2)\big)$, for all $x=(x_1,x_2)\in\bbZ^2$, which 
are defined by choosing any reference point $\overline{x}_2\in \bbZ$ and setting
\begin{flalign}\label{eqn:hatT}
\widehat{T}(x)\,:=\,\begin{cases}
\oone ~&~~\text{for }x_2\,=\,\overline{x}_2\quad,\\
T_2(x_1, x_2-1)\cdots T_2(x_1,\overline{x}_2+1) \,T_2(x_1,\overline{x}_2) ~&~~\text{for }x_2\,>\,\overline{x}_2\quad,\\
T_2(x_1,x_2)^{-1}\cdots T_2(x_1,\overline{x}_2-2)^{-1}\,T_2(x_1,\overline{x}_2-1)^{-1} ~&~~\text{for }x_2\,<\,\overline{x}_2\quad.\\
\end{cases}
\end{flalign}
Note that these elements describe the parallel transport 
along the $x_2$-direction from the reference point $(x_1,\overline{x}_2)$ 
to the point $x=(x_1,x_2)$. Under the coaction \eqref{eqn:rhodCrit} of gauge transformations, 
these elements transform as
\begin{flalign}\label{eqn:hatTtransformation}
\rho\big(\widehat{T}(x)\big)\,=\,U(x)^{-1}\,\widehat{T}(x)\,U(x_1,\overline{x}_2)\quad.
\end{flalign}
We define a morphism $\pi : Z(\bbZ^2)\to Z^{\gf}(\bbZ^2)$ of derived affine schemes
by setting
\begin{flalign}
\pi^\ast\,:\,\O\big(Z^{\gf}(\bbZ^2)\big)~&\longrightarrow~\O\big( Z(\bbZ^2)\big)\quad,\\
\nn T(x)~&\longmapsto~\widehat{T}(x+e_1)^{-1}\,T_1(x)\,\widehat{T}(x)\quad,\\
\nn \xi_i(x)~&\longmapsto~\widehat{T}(x)^{-1}\,\xi_i(x)\,\widehat{T}(x)\quad,\\
\nn \xi(x)~&\longmapsto~\widehat{T}(x)^{-1}\,\xi(x)\,\widehat{T}(x)\quad,
\end{flalign}
and a morphism $\und{\pi} : \G(\bbZ^2)\to \G^{\gf}(\bbZ^2)$ of affine group schemes
by setting
\begin{flalign}
\und{\pi}^\ast\,:\,\O\big(\G^{\gf}(\bbZ^2)\big)~\longrightarrow~\O\big( \G(\bbZ^2)\big)~~,\quad
U(x_1)~\longmapsto~ U(x_1,\overline{x}_2)\quad.
\end{flalign}
Using \eqref{eqn:hatTtransformation}, one directly checks that $\pi$ is equivariant
relative to $\und{\pi}$, i.e.\ the diagram
\begin{flalign}
\begin{gathered}
\xymatrix{
\ar[d]_-{\pi\times \und{\pi}}Z(\bbZ^2)\times \G(\bbZ^2)\ar[r]^-{r} ~&~Z(\bbZ^2)\ar[d]^-{\pi}\\
Z^{\gf}(\bbZ^2)\times \G^{\gf}(\bbZ^2) \ar[r]_-{r}~&~Z^{\gf}(\bbZ^2)
}
\end{gathered}
\end{flalign}
in $\dAff$ commutes. This allows us to define a morphism 
\begin{flalign}\label{eqn:Pimapsimplicial}
\begin{gathered}
\xymatrix{
Z^{\gf}(\bbZ^2) \ar[r]~&~\ar@<1ex>[l]\ar@<-1ex>[l]
Z^{\gf}(\bbZ^2) \times \G^{\gf}(\bbZ^2)
\ar@<1ex>[r] \ar@<-1ex>[r]~&~\ar@<2ex>[l]\ar@<-2ex>[l]\ar[l] 
Z^{\gf}(\bbZ^2)\times \G^{\gf}(\bbZ^2)^{2}~~\cdots\\
\ar[u]^-{\pi} Z(\bbZ^2) \ar[r]~&~\ar@<1ex>[l]\ar@<-1ex>[l] \ar[u]^-{\pi\times\und{\pi}} 
Z(\bbZ^2)\times \G(\bbZ^2) 
\ar@<1ex>[r] \ar@<-1ex>[r]~&~\ar@<2ex>[l]\ar@<-2ex>[l]\ar[l] 
\ar[u]^-{\pi\times\und{\pi}^2} Z(\bbZ^2)\times \G(\bbZ^2)^{2}~~\cdots
}
\end{gathered}
\end{flalign}
of simplicial diagrams in $\dAff$
which goes in the opposite direction of \eqref{eqn:Jmapsimplicial}.
\begin{propo}\label{prop:axialgauge}
The two morphisms of simplicial diagrams in $\dAff$ given in
\eqref{eqn:Jmapsimplicial} and \eqref{eqn:Pimapsimplicial}
are quasi-inverse to each other. Hence, by Lemma \ref{lem:simplicialhomotopy}
the induced morphism \eqref{eqn:Jmap} between the associated derived 
quotient stacks is a weak equivalence in the model category $\dSt$.
\end{propo}
\begin{proof}
One directly verifies by using the above formulas  that $\pi\,j=\id_{Z^{\gf}(\bbZ^2)}$
and $\und{\pi}\, \,\und{j} = \id_{\G^{\gf}(\bbZ^2)}$, so the composition
of \eqref{eqn:Jmapsimplicial} followed by \eqref{eqn:Pimapsimplicial} is the identity.
\sk

For the other composition $j\,\pi : Z(\bbZ^2)\to Z(\bbZ^2)$, one finds
by using the above formulas that
\begin{subequations}\label{eqn:TMPformulas}
\begin{flalign}
\pi^\ast\,j^\ast\big(T_1(x)\big)\,&=\,\widehat{T}(x+e_1)^{-1}\, T_1(x)\,\widehat{T}(x)\quad,\\
\pi^\ast\,j^\ast\big(T_2(x)\big)\,&=\,\oone\,=\,\widehat{T}(x+e_2)^{-1}\, T_2(x)\,\widehat{T}(x)\quad,\\
\pi^\ast\,j^\ast\big(\xi_i(x)\big)\,&=\,\widehat{T}(x)^{-1}\, \xi_i(x)\,\widehat{T}(x)\quad,\\
\pi^\ast\,j^\ast\big(\xi(x)\big)\,&=\,\widehat{T}(x)^{-1}\, \xi(x)\,\widehat{T}(x)\quad,
\end{flalign}
\end{subequations}
while for $\und{j}\,\und{\pi} : \G(\bbZ^2)\to \G(\bbZ^2)$ one finds
\begin{flalign}
\und{\pi}^\ast\,\und{j}^\ast\big(U(x)\big)\,=\,U(x_1,\overline{x}_2)\quad.
\end{flalign}
Defining the $\dAff$-morphism 
$\hat{\eta} := (\id_{Z(\bbZ^2)},\eta) : Z(\bbZ^2)\to Z(\bbZ^2)\times \G(\bbZ^2)$ by
\begin{flalign}
\eta^\ast\,:\,\O\big(\G(\bbZ^2)\big)~\longrightarrow~\O\big(Z(\bbZ^2)\big)~~,\quad
U(x)~\longmapsto~\widehat{T}(x)\quad,
\end{flalign}
we observe that \eqref{eqn:TMPformulas} implies that the diagrams
\begin{flalign}
\begin{gathered}
\xymatrix{
\ar[dr]_-{\id_{Z(\bbZ^2)}} Z(\bbZ^2) \ar[r]^-{\hat{\eta}}~&~ Z(\bbZ^2)\times \G(\bbZ^2)\ar[d]^-{\pr_{Z(\bbZ^2)}} ~&~ 
\ar[dr]_-{j\,\pi}Z(\bbZ^2) \ar[r]^-{\hat{\eta}} ~&~ Z(\bbZ^2)\times \G(\bbZ^2)\ar[d]^-{r}\\
~&~Z(\bbZ^2) ~&~ ~&~Z(\bbZ^2)
}
\end{gathered}
\end{flalign}
in $\dAff$ commute. Hence, $\hat{\eta} : \id_{Z(\bbZ^2)/\G(\bbZ^2)} \Rightarrow j\,\pi:
Z(\bbZ^2)/\G(\bbZ^2)\to Z(\bbZ^2)/\G(\bbZ^2)$ 
is a candidate for a natural isomorphism of
functors between groupoid objects in $\dAff$. 
(See e.g.\ \cite{Roberts} for a brief summary and the relevant definitions
of internal category theory.)
It remains to verify naturality of $\hat{\eta}$, which amounts to checking
that the diagram
\begin{flalign}
\begin{gathered}
\xymatrix@C=3.5em{
\ar[d]_-{\text{flip}} 
Z(\bbZ^2)\times \G(\bbZ^2) \ar[rrr]^-{\eta\times \und{j}\,\und{\pi}}~&~~&~~&~\G(\bbZ^2)\times \G(\bbZ^2)\ar[d]^-{m}\\
\G(\bbZ^2)\times Z(\bbZ^2) \ar[rr]_-{(\pr_{\G(\bbZ^2)},\,\eta\,r)}~&~~&~\G(\bbZ^2)\times \G(\bbZ^2)\ar[r]_-{m}~&~\G(\bbZ^2)
}
\end{gathered}
\end{flalign}
in $\dAff$ commutes, where by $m$ we denote the group multiplication of $\G(\bbZ^2)$.
Passing over to the opposite $\dgCAlg^{\leq 0}$-morphisms, we compute for the upper 
path in this diagram
\begin{subequations}
\begin{flalign}
U(x)\,\longmapsto\, U(x)\otimes U(x)\,\longmapsto\, \widehat{T}(x)\,U(x_1,\overline{x}_2)
\end{flalign}
and for the lower path we find
\begin{flalign}
U(x)\,\longmapsto\, U(x)\otimes U(x)\,\longmapsto\,U(x)\, \rho\big(\widehat{T}(x)\big)
\,=\,\widehat{T}(x)\,U(x_1,\overline{x}_2)\quad,
\end{flalign}
\end{subequations}
where in the last step we used the property \eqref{eqn:hatTtransformation}.
This shows that $\hat{\eta}$ is indeed a natural isomorphism. The proof of this
proposition then follows by applying the nerve functor to obtain a simplicial homotopy
between the identity and the composition of \eqref{eqn:Pimapsimplicial} followed by \eqref{eqn:Jmapsimplicial}.
\end{proof}


\section{\label{sec:localconstancy}Local derived critical loci and their functorial structure}
In this section we extract suitable local data $\Sol(V)\in\dSt$ from the derived critical locus 
$\dCrit(S)\simeq \big[ Z(\bbZ^2)/ \G(\bbZ^2) \big]$ from Section \ref{sec:dCrit}
which is supported in rectangular subsets
\begin{subequations}\label{eqn:Vsubset}
\begin{flalign}\label{eqn:Vsubset1}
V\,=\, [a,b]\times [c,d]\,\subseteq\,\bbZ^2
\end{flalign} 
of the discrete spacetime $\bbZ^2$, where $[a,b] := \{a,a+1,\dots,b\}\subseteq \bbZ$ 
and $[c,d] := \{c,c+1,\dots,d\}\subseteq \bbZ$ denote discrete intervals. 
Note that we also allow for unbounded intervals,
e.g.\ $a,c=-\infty$ and $b,d=+\infty$ are admissible. 
We shall always assume that both sides of such rectangular subsets 
are of length $\geq 2$, i.e.\ we demand that
\begin{flalign}\label{eqn:Vsubset2}
b-a\,\geq 2\quad\text{and}\quad d-c\geq 2\quad.
\end{flalign}
\end{subequations}
(This is equivalent to demanding that both the discrete intervals
$[a,b]$ and $[c,d]$ contain at least three points.)
These local data will assemble into a functor
$\Sol : \Rect(\bbZ^2)^\op\to \dSt$ from the opposite of the category $\Rect(\bbZ^2)$
consisting of all rectangular subsets $V\subseteq \bbZ^2$ with both sides of length $\geq 2$ and 
morphisms $\iota_V^{V^\prime}:V\to V^\prime$ given by subset inclusions $V\subseteq V^\prime$. 
We shall prove that this functor is locally constant in the sense 
that $\Sol\big(\iota_V^{V^\prime}\big) :\Sol(V^\prime)\to \Sol(V)$
is a weak equivalence of derived stacks for every morphism $\iota_V^{V^\prime}:V\to V^\prime$ in $\Rect(\bbZ^2)$.
\sk

To extract the local data in $\dCrit(S)\simeq \big[ Z(\bbZ^2)/ \G(\bbZ^2) \big]$
which is supported in a rectangular subset \eqref{eqn:Vsubset},
one has to carefully pay attention
to the fact that connections and finite-difference operators on $\bbZ^2$
are extended objects which do not preserve supports. This forces us
to demand different support conditions for the various generators of 
\eqref{eqn:OGau} and \eqref{eqn:OZYM}. A consistent choice is given as follows:
We define
\begin{flalign}\label{eqn:OGV}
\O\big(\G(V)\big)\,\subseteq \, \O\big(\G(\bbZ^2)\big)
\end{flalign}
to be the commutative Hopf subalgebra which is generated by 
$U(x)$, for all $x \in [a,b]\times [c,d]$.
We further define
\begin{flalign}\label{eqn:OZV}
\O\big(Z(V)\big)\,\subseteq\,\O\big(Z(\bbZ^2)\big)
\end{flalign}
to be the graded commutative subalgebra which is generated by
\begin{itemize}
\item $T_1(x)$, for all $x\in [a,b-1]\times [c,d]$,
\item $T_2(x)$, for all $x\in [a,b]\times [c,d-1]$,
\item $\xi_1(x)$, for all $x\in [a,b-1]\times [c+1,d-1]$,
\item $\xi_2(x)$, for all $x\in [a+1,b-1]\times [c,d-1]$, and 
\item $\xi(x)$, for all $x\in[a+1,b-1]\times [c+1,d-1]$.
\end{itemize}
It is helpful to visualize these support restrictions 
($V$ consists of all points in the red rectangle):
\begin{flalign}\label{eqn:supportrestrictions}
\begin{gathered}
\begin{tikzpicture}
\foreach \i in {0,...,5}
\foreach \j in {0,...,5}
\filldraw (\i/2,\j/2) circle (1.5pt);
\foreach \i in {0,...,5}
\foreach \j in {0,...,5}
\filldraw (5+\i/2,\j/2) circle (1.5pt);
\foreach \i in {0,...,5}
\foreach \j in {0,...,5}
\filldraw (10+\i/2,\j/2) circle (1.5pt);
\foreach \i in {0,...,5}
\foreach \j in {0,...,5}
\filldraw (\i/2,5+\j/2) circle (1.5pt);
\foreach \i in {0,...,5}
\foreach \j in {0,...,5}
\filldraw (5+\i/2,5+\j/2) circle (1.5pt);
\foreach \i in {0,...,5}
\foreach \j in {0,...,5}
\filldraw (10+\i/2,5+\j/2) circle (1.5pt);
\draw[thick,->] (-1,3.75) -- node[below]{{\footnotesize $x_1$}} (0,3.75);
\draw[thick,->] (-1,3.75)-- node[left]{{\footnotesize $x_2$}} (-1,4.75);
\draw[red] (1/2 -0.2,5+1/2 -0.2) -- (2 +0.2,5+1/2 -0.2) -- (2 +0.2,7 +0.2) -- (1/2 -0.2,7 +0.2) -- 
(1/2 -0.2 ,5+1/2 -0.2);
\draw[red] (5+ 1/2 -0.2,5+1/2 -0.2) -- (5+ 2 +0.2,5+1/2 -0.2) -- (5+ 2 +0.2,7 +0.2) -- (5+ 1/2 -0.2,7 +0.2) -- 
(5+ 1/2 -0.2 ,5+1/2 -0.2);
\draw[red] (10+ 1/2 -0.2,5+1/2 -0.2) -- (10+ 2 +0.2,5+1/2 -0.2) -- (10+ 2 +0.2,7 +0.2) -- (10+ 1/2 -0.2,7 +0.2) -- 
(10+ 1/2 -0.2 ,5+1/2 -0.2);
\draw[red] (1/2 -0.2,1/2 -0.2) -- (2 +0.2,1/2 -0.2) -- (2 +0.2,2 +0.2) -- (1/2 -0.2,2 +0.2) -- 
(1/2 -0.2 ,1/2 -0.2);
\draw[red] (5+ 1/2 -0.2,1/2 -0.2) -- (5+ 2 +0.2,1/2 -0.2) -- (5+ 2 +0.2,2 +0.2) -- (5+ 1/2 -0.2,2 +0.2) -- 
(5+ 1/2 -0.2 ,1/2 -0.2);
\draw[red] (10+ 1/2 -0.2,1/2 -0.2) -- (10+ 2 +0.2,1/2 -0.2) -- (10+ 2 +0.2,2 +0.2) -- (10+ 1/2 -0.2,2 +0.2) -- 
(10+ 1/2 -0.2 ,1/2 -0.2);
\fill[gray,opacity=0.2] (1/2 -0.2,5+1/2 -0.2) -- (2 +0.2,5+1/2 -0.2) -- (2 +0.2,7 +0.2) -- (1/2 -0.2,7 +0.2) -- 
(1/2 -0.2 ,5+1/2 -0.2);
\fill[gray,opacity=0.2] (5+ 1/2 -0.2,5+1/2 -0.2) -- (5+ 2 +0.2 -0.5 ,5+1/2 -0.2) -- (5+ 2 +0.2  -0.5,7 +0.2) -- (5+ 1/2 -0.2,7 +0.2) -- 
(5+ 1/2 -0.2 ,5+1/2 -0.2);
\fill[gray,opacity=0.2] (10+ 1/2 -0.2,5+1/2 -0.2) -- (10+ 2 +0.2,5+1/2 -0.2) -- (10+ 2 +0.2,7 +0.2  -0.5) -- (10+ 1/2 -0.2,7 +0.2  -0.5) -- 
(10+ 1/2 -0.2 ,5+1/2 -0.2);
\node at (1.25, 8) {$U(x)$};
\node at (5+1.25, 8) {$T_1(x)$};
\node at (10+1.25, 8) {$T_2(x)$};
\fill[gray,opacity=0.2] (1/2 -0.2,1 -0.2) -- (2 +0.2-1/2,1 -0.2) -- (2 +0.2 -1/2,2 +0.2-1/2) -- (1/2 -0.2,2 +0.2-1/2) -- 
(1/2 -0.2 ,1 -0.2);
\fill[gray,opacity=0.2] (5+ 1/2+1/2 -0.2,1/2 -0.2) -- (5+ 2 +0.2 -0.5 ,1/2 -0.2) -- (5+ 2 +0.2  -0.5,2 +0.2-1/2) -- (5+ 1/2+1/2 -0.2,2 +0.2-1/2) -- 
(5+ 1/2 +1/2-0.2 ,1/2 -0.2);
\fill[gray,opacity=0.2] (10+ 1/2+1/2 -0.2,1/2 -0.2+1/2) -- (10+ 2 +0.2 -1/2,1/2 -0.2+1/2) -- (10+ 2 +0.2-1/2,2 +0.2  -0.5) -- (10+ 1/2 -0.2 +1/2,2 +0.2  -0.5) -- 
(10+ 1/2 +1/2-0.2 ,1/2 +1/2-0.2);
\node at (1.25, 3) {$\xi_1(x)$};
\node at (5+1.25, 3) {$\xi_2(x)$};
\node at (10+1.25, 3) {$\xi(x)$};
\end{tikzpicture}
\end{gathered}
\end{flalign}
One directly checks that the differential \eqref{eqn:OZYM}
closes on these generators, hence \eqref{eqn:OZV} is a commutative dg-subalgebra.
Furthermore, one shows that the coaction \eqref{eqn:rhodCrit} restricts
to $\rho : \O\big(Z(V)\big)\to \O\big(Z(V)\big)\otimes \O\big(\G(V)\big)$. Hence, we can define
the derived quotient stack
\begin{flalign}
\Sol(V)\,:=\, \big[Z(V)/\G(V)\big]\,\in\,\dSt\quad,
\end{flalign}
for each object $V\in \Rect(\bbZ^2)$, which captures the local data
of $\dCrit(S)\simeq \big[Z(\bbZ^2)/\G(\bbZ^2)\big]$ that is supported
in $V\subseteq \bbZ^2$.
\sk

Given any morphism $\iota_V^{V^\prime} : V\to V^\prime$ in $\Rect(\bbZ^2)$,
one evidently has that 
\begin{subequations}
\begin{flalign}
\O\big(\G(V)\big)\,\subseteq\,\O\big(\G(V^\prime)\big)
\end{flalign}
is a commutative Hopf subalgebra and that
\begin{flalign}
\O\big(Z(V)\big)\,\subseteq\,\O\big(Z(V^\prime)\big)
\end{flalign}
\end{subequations}
is a commutative dg-subalgebra. These inclusions are compatible with the coactions,
hence we obtain a $\dSt$-morphism
\begin{flalign}
\Sol\big(\iota_V^{V^\prime}\big)\,:\,\Sol(V^\prime)\,=\,\big[Z(V^\prime)/\G(V^\prime)\big]~\longrightarrow~\Sol(V)\,=\,\big[Z(V)/\G(V)\big]
\end{flalign}
which describes the restriction of local data along the $\Rect(\bbZ^2)$-morphism
$\iota_V^{V^\prime} : V\to V^\prime$. This defines a functor
\begin{flalign}\label{eqn:Solfunctor}
\Sol\,:\,\Rect(\bbZ^2)^\op~\longrightarrow~\dSt
\end{flalign}
to the model category of derived stacks.
\begin{theo}\label{theo:locallyconstant}
The functor \eqref{eqn:Solfunctor} is locally constant
in the sense that $\Sol\big(\iota_V^{V^\prime}\big) :\Sol(V^\prime)\to \Sol(V)$
is a weak equivalence in the model category $\dSt$,
for every morphism $\iota_V^{V^\prime} : V\to V^\prime$ in $\Rect(\bbZ^2)$.
\end{theo}
\begin{proof}
We start by observing that every morphism 
$\iota_V^{V^\prime}: V\to  V^\prime$ in $\Rect(\bbZ^2)$
admits a factorization
\begin{flalign}
V\,=\,[a,b]\times [c,d]~\longrightarrow~[a,b]\times [c^\prime,d^\prime]~
\longrightarrow~ [a^\prime,b^\prime]\times[c^\prime,d^\prime]\,=\,V^\prime
\end{flalign}
into an interval inclusion along the $x_2$-direction
and an interval inclusion along the $x_1$-direction.
Since the class of weak equivalences is stable under compositions,
it suffices to prove that the functor $\Sol:\Rect(\bbZ^2)^\op \to \dSt$ 
assigns a weak equivalence in $\dSt$ to 
each of these more basic morphisms. Leveraging the symmetry (up to signs)
of \eqref{eqn:OGau}, \eqref{eqn:OZYM} and \eqref{eqn:rhodCrit} under exchanging $x_1$ and $x_2$,
we can restrict our attention to $x_2$-interval inclusions 
$\iota_V^{V^\prime}: V=[a,b]\times[c,d] \to  V^\prime=[a,b]\times [c^\prime,d^\prime]$. 
\sk

Recalling from Subsection \ref{subsec:axialgauge}
the derived critical locus $\dCrit^{\gf}(S)=\big[Z^{\gf}(\bbZ^2)/\G^{\gf}(\bbZ^2)\big]$ in axial gauge $T_2(x)=\oone$,
for all $x\in\bbZ^2$,  we can extract with the same support
conditions as in \eqref{eqn:supportrestrictions} its local data
and obtain a $\dSt$-morphism 
$\Sol^{\gf}\big(\iota_V^{V^\prime}\big):
\Sol^{\gf}(V^\prime) =\big[Z^{\gf}(V^\prime)/\G^{\gf}(V^\prime)\big]\to \Sol^{\gf}(V)=\big[Z^{\gf}(V)/\G^{\gf}(V)\big]$.
The $\dSt$-morphism in \eqref{eqn:Jmap} restricts to local data and yields a commutative square
\begin{flalign}\label{eqn:squaregf}
\begin{gathered}
\xymatrix@C=4em{
\ar[d]_-{J_{V^\prime}}^-{\sim} \Sol^{\gf}(V^\prime) \ar[r]^-{\Sol^{\gf}(\iota_V^{V^\prime})}~&~\Sol^{\gf}(V) \ar[d]^-{J_{V}}_-{\sim} \\
\Sol(V^\prime) \ar[r]_-{\Sol(\iota_V^{V^\prime})}~&~\Sol(V)
}
\end{gathered}
\end{flalign}
in $\dSt$. The vertical arrows are weak equivalences in $\dSt$
because the proof of Proposition \ref{prop:axialgauge} applies locally to $V\subseteq \bbZ^2$
and $V^\prime\subseteq \bbZ^2$. Using $2$-out-of-$3$ for weak equivalences,
our problem of proving that $\Sol\big(\iota_V^{V^\prime}\big)$
is a weak equivalence is equivalent to showing that $\Sol^{\gf}\big(\iota_V^{V^\prime}\big)$ 
is one. Using that our morphism $\iota_V^{V^\prime} : V\to V^\prime$
increases the rectangular subset only in the $x_2$-direction
and recalling that the gauge group in axial gauge \eqref{eqn:Gaugf} is independent of
$x_2$, we find that
\begin{flalign}
\O\big(\G^{\gf}(V)\big)\,=\,\O\big(\G^{\gf}(V^\prime)\big)\,=\, \bigotimes_{x_1\in[a,b]}\O(\GL_n)\quad.
\end{flalign}
Hence, if we can prove that the inclusion
\begin{flalign}\label{eqn:gfinclusion}
\O\big(Z^{\gf}(V)\big)\,\subseteq\,\O\big(Z^{\gf}(V^\prime)\big)
\end{flalign}
is a weak equivalence in $\dgCAlg^{\leq 0}$, it would
follow that the induced morphism $\Sol^{\gf}\big(\iota_V^{V^\prime}\big):\Sol^{\gf}(V^\prime) \to 
\Sol^{\gf}(V)$ between homotopy colimits \eqref{eqn:X/Ghocolim} 
is a weak equivalence in $\dSt$.
\sk

Our strategy is to break the problem of proving that \eqref{eqn:gfinclusion} is a weak equivalence
into smaller steps.  
Every $x_2$-interval inclusion $[a,b]\times [c,d]\subseteq [a,b]\times [c^\prime,d^\prime]$ 
can be presented as a (possibly transfinite) composition of primitive inclusions of two types: 
The first type increases the right endpoint $d\mapsto d+1$
by one step and the second type decreases the left endpoint $c\mapsto c-1$ by one step. 
Recalling that quasi-isomorphisms are closed under transfinite compositions, 
it suffices to prove that \eqref{eqn:gfinclusion} 
is a weak equivalence for any primitive inclusion. Furthermore, 
since the proofs for the two types of primitive inclusions are 
similar, it suffices to consider only one of them.
\sk

Using the above observations, we consider in what follows a $\Rect(\bbZ^2)$-morphism
\begin{flalign}
\iota_V^{V^\prime} \,:\, V=[a,b]\times[c,d]~\longrightarrow~ V^\prime\,=\,[a,b]\times [c,d+1]
\end{flalign}
which increases the right $x_2$-interval endpoint by one step.
Let us denote by
\begin{subequations}\label{eqn:AandA'}
\begin{flalign}
A\,&:=\, \O\big(Z^\gf(V)\big)^0\,=\,\bigotimes_{x\in [a,b-1]\times[c,d]}\O(\GL_n)\,\subseteq\, \O\big(Z^\gf(V)\big)
\quad,\\[4pt]
A^\prime\,&:=\, \O\big(Z^\gf(V^\prime)\big)^0\,=\,\bigotimes_{x\in [a,b-1]\times[c,d+1]}
\O(\GL_n)\,\subseteq\, \O\big(Z^\gf(V^\prime)\big)
\end{flalign}
\end{subequations}
the commutative subalgebras consisting of all elements of degree $0$.
Note that these are also commutative dg-subalgebras (with trivial differential)
because $\O\big(Z^\gf(V)\big),\O\big(Z^\gf(V^\prime)\big)\in \dgCAlg^{\leq 0}$
are concentrated in non-positive degrees. The inclusion $A\subseteq A^\prime$ is
clearly not a weak equivalence because $A$ and $A^\prime$ are discrete and $A^\prime$ 
contains additional elements which are supported in 
$[a,b-1]\times \{d+1\}\subseteq \bbZ^2$. To remedy this issue,
we introduce a bigger commutative  dg-subalgebra 
\begin{flalign}\label{eqn:widetildeA'}
\widetilde{A}^\prime\,:=\,\!\! \bigotimes_{x\in [a,b-1]\times[c,d+1]}\!\!\O(\GL_n) \otimes \!\!\bigotimes_{x\in [a,b-1]\times\{d\}}\!\! \Sym\big(\gl_n[1]\big)\, \subseteq\, \O\big(Z^\gf(V^\prime)\big)
\end{flalign}
which further includes the degree $-1$ generators $\xi_1(x_1,d)\in \O\big(Z^\gf(V^\prime)\big)$,
for all $x_1\in[a,b-1]$, that are \textit{not} contained in 
$\O\big(Z^\gf(V)\big)\subseteq \O\big(Z^\gf(V^\prime)\big)$. (Recall the 
support conditions from \eqref{eqn:supportrestrictions}.) We show in Appendix
\ref{app:details} that the inclusion $A\subseteq \widetilde{A}^\prime$
is a weak equivalence in $\dgCAlg^{\leq 0}$. Let us further observe that there exists
a retraction
\begin{flalign}\label{eqn:retract}
\begin{gathered}
\xymatrix{
\ar@/_1.5pc/[rr]_-{\id_A}A\ar[r]^-{\subseteq }~&~\widetilde{A}^\prime\ar[r]^-{r}~&~A
}
\end{gathered}
\end{flalign}
defined by the $\dgCAlg^{\leq 0}$-morphism 
\begin{flalign}\label{eqn:retractformula}
r\,:\, \widetilde{A}^\prime~&\longrightarrow~A\quad,\\
\nn A\ni a~&\longmapsto~ a\quad,\\
\nn T(x_1,d+1)~&\longmapsto~ T(x_1,d)\,E^\gf(x_1,d-1)^{-1} \,=\,T(x_1,d)\,T(x_1,d-1)^{-1}\,T(x_1,d) \quad,\\
\nn \xi_1(x_1,d)~&\longmapsto~0\quad.
\end{flalign}
Since $A\subseteq \widetilde{A}^\prime$
is a weak equivalence it follows from \eqref{eqn:retract} that 
$r : \widetilde{A}^\prime\to A$ is one too.
\sk

With these preparations, we can derive 
equivalent but simpler characterizations 
for \eqref{eqn:gfinclusion} being a weak equivalence.
From a change-of-base along the weak equivalence $A\subseteq \widetilde{A}^\prime$,
we obtain a commutative triangle
\begin{flalign}
\begin{gathered}
\xymatrix{
\widetilde{A}^\prime\otimes _{A}\O\big(Z^\gf(V)\big) \ar[r]~&~\O\big(Z^\gf(V^\prime)\big)\\
\O\big(Z^\gf(V)\big)  \ar[u]^-{\sim}\ar[ru]_-{\subseteq}~&~ 
}
\end{gathered}
\end{flalign}
in $\dgCAlg^{\leq 0}$. The left vertical arrow is a weak equivalence
because $\O\big(Z^\gf(V)\big)$ is a semi-free extension of $A$, which
by left properness of $\dgCAlg^{\leq 0}$ (see e.g.\ \cite[Corollary 3.4]{Manetti}) 
implies that $(-)\otimes_A\O\big(Z^\gf(V)\big)$
preserves weak equivalences. Hence, by $2$-out-of-$3$ we can equivalently prove 
that the top horizontal arrow is a weak equivalence.
Applying to the top horizontal arrow a change-of-base along the retraction $r : \widetilde{A}^\prime\to A$
(which is a weak equivalence too) yields a commutative diagram
\begin{flalign}
\begin{gathered}
\xymatrix@C=1.5em{
\ar[d]_-{\sim}\widetilde{A}^\prime\otimes _{A}\O\big(Z^\gf(V)\big) \ar[rr]~&~~&~\O\big(Z^\gf(V^\prime)\big)\ar[d]^-{\sim}\\
\ar[dr]_-{\cong} A\otimes_{\widetilde{A}^\prime}\widetilde{A}^\prime\otimes _{A}\O\big(Z^\gf(V)\big) \ar[rr]~&~~&~
A\otimes_{\widetilde{A}^\prime}\O\big(Z^\gf(V^\prime)\big)\\
~&~ \O\big(Z^\gf(V)\big)\ar[ru]_-{k}~&~
}
\end{gathered}
\end{flalign}
in $\dgCAlg^{\leq 0}$. The left and right vertical arrows are weak equivalences
because $\widetilde{A}^\prime\otimes_{A}\O\big(Z^\gf(V)\big)$ and 
$\O\big(Z^\gf(V^\prime)\big)$ are semi-free extensions of $\widetilde{A}^\prime$. The downward-right
pointing arrow is an isomorphism because of the retraction property \eqref{eqn:retract}.
Hence, by $2$-out-of-$3$ we can equivalently prove 
that the upward-right pointing arrow labeled by $k$ is a weak equivalence.
\sk

By direct inspection, one observes that the $\dgCAlg^{\leq 0}$-morphism 
$k :\O\big(Z^\gf(V)\big)\to A\otimes_{\widetilde{A}^\prime}\O\big(Z^\gf(V^\prime)\big)$
is a semi-free extension by the generators (recall the support properties \eqref{eqn:supportrestrictions})
\begin{itemize}
\item $\xi_2(x_1,d)$, for all $x_1\in [a+1,b-1]$, and
\item $\xi(x_1,d)$, for all $x_1\in [a+1,b-1]$.
\end{itemize}
Using the explicit formulas for the retraction \eqref{eqn:retractformula}
and the differential \eqref{eqn:OZYMtg} of $\O\big(Z^\gf(V^\prime)\big)$, 
one computes the differential of these generators in 
$A\otimes_{\widetilde{A}^\prime}\O\big(Z^\gf(V^\prime)\big)$ 
and finds that
\begin{flalign}
\dd\xi_2(x_1,d)\,=\, \dd\xi_2(x_1,d-1)\quad,\qquad
\dd\xi(x_1,d)\,=\,-\xi_2(x_1,d) + \xi_2(x_1,d-1)\quad.
\end{flalign}
It follows that there exists a retraction of $k$ given by the $\dgCAlg^{\leq 0}$-morphism
\begin{flalign}
q\, :\, A\otimes_{\widetilde{A}^\prime}\O\big(Z^\gf(V^\prime)\big)~&\longrightarrow~
\O\big(Z^\gf(V)\big)\quad,\\
\nn \O\big(Z^\gf(V)\big)\ni a~&\longmapsto~a\quad,\\
\nn \xi_2(x_1,d)~&\longmapsto~\xi_2(x_1,d-1)\\
\nn \xi(x_1,d)~&\longmapsto~0\quad,
\end{flalign}
i.e.\ $q\,k=\id$. This is further a deformation retraction $\partial(h) = \id - k\,q$ for the 
$\O\big(Z^\gf(V)\big)$-linear homotopy $h$ which is defined on the relative generators by
\begin{flalign}
h\big(\xi_2(x_1,d)\big)\,=\,-\xi(x_1,d)\quad,\qquad h\big(\xi(x_1,d)\big)\,=\,0\quad,
\end{flalign}
and extended to the semi-free extension $A\otimes_{\widetilde{A}^\prime}\O\big(Z^\gf(V^\prime)\big)$
of $\O\big(Z^\gf(V)\big)$ via the usual symmetric tensor trick, see e.g.\ \cite{Berglund}.
This completes the proof that $k$ is a weak equivalence.
\end{proof}


\section{\label{sec:PFA}$\mathbf{dgCat}$-valued prefactorization algebra of classical observables}
The aim of this section is to construct out of the local derived
critical loci from Section \ref{sec:localconstancy} a locally constant
prefactorization algebra on the discrete spacetime $\bbZ^2$
which takes values in dg-categories. Category-valued prefactorization
algebras appeared before in the works of Ben-Zvi, Brochier and Jordan \cite{Jordan1,Jordan2}
in the context of representation theory 
and they have been proposed as a categorification of algebraic quantum field 
theory in \cite{BPSW,BSchapter}. The kind of prefactorization algebras
on $\bbZ^2$ we will consider below are encoded by the following operad.
\begin{defi}\label{def:PFA}
The \textit{rectangular prefactorization operad} $\P_{\bbZ^2}$ 
on the square lattice $\bbZ^2$ is defined as the 
following colored symmetric operad:
\begin{itemize}
\item An object in $\P_{\bbZ^2}$ is a bounded rectangular subset
$V = [a,b]\times [c,d]\subseteq \bbZ^2$ with both sides of length $\geq 2$.

\item There exists exactly one $n$-ary operation
$\iota_{\und{V}}^V : \und{V} := (V_1,\dots,V_n)\to V$ if $V_i\subseteq V$, for all $i=1,\dots,n$,
and $V_i\cap V_j =\varnothing$, for all $i\neq j=1,\dots,n$. In particular,
there exists a unique operation $\varnothing\to V$ of arity zero for each object $V$.
\end{itemize}
Operadic composition is forced by these definitions and the operadic
units are the $1$-ary operations $\iota_V^V:V\to V$ associated with the identities 
$V=V$. The permutation action $\big(\iota_{\und{V}}^V: \und{V} \to V\big)
\mapsto \big(\iota_{\und{V}\sigma}^V : \und{V}\sigma \to V\big)$
is defined by permuting tuples $\und{V}\sigma = (V_{\sigma(1)},\dots,V_{\sigma(n)})$, for
all $\sigma\in\Sigma_n$.
\end{defi}
\begin{rem}\label{rem:bounded}
In contrast to the category $\Rect(\bbZ^2)$ from Section \ref{sec:localconstancy},
we demand in the definition of the prefactorization operad $\P_{\bbZ^2}$
that all subsets $V = [a,b]\times [c,d]\subseteq \bbZ^2$ are bounded, i.e.\ $V\subseteq \bbZ^2$
only consists of finitely many points. This finiteness property allows us to 
provide a simple and explicit model \eqref{eqn:Fonobjects} for the dg-category $\FFF(V)$
assigned by our prefactorization algebra to $V$.
\end{rem}

A (rectangular) prefactorization algebra on $\bbZ^2$ is then defined
as a pseudo-multifunctor $\FFF : \P_{\bbZ^2}\to \dgCat$
to the symmetric monoidal $2$-category $\dgCat$ of dg-categories,
dg-functors and dg-natural transformations. (See e.g.\
\cite[Section 2.3]{Keller} or \cite[Section 1.4]{Kelly}
for an explicit description of the symmetric monoidal structure.) 
On objects and $1$-ary operations in $\P_{\bbZ^2}$, 
we define our prefactorization algebra $\FFF$  by composing the
functor $\Sol : \Rect(\bbZ^2) \to \dSt^\op$ from \eqref{eqn:Solfunctor},
which assigns the local derived critical loci of our lattice Yang-Mills model,
with the pseudo-functor $\QCoh : \dSt^\op\to \dgCat$ from \eqref{eqn:QCohdSt}
which assigns dg-categories of quasi-coherent complexes. 
Recalling that $\Sol(V) = \big[Z(V)/\G(V)\big]\in\dSt$ is a derived quotient
stack with $\G(V) = \prod_{x\in V}\GL_n$ a reductive affine group scheme 
(since $V\subseteq \bbZ^2$ is a finite subset),
we obtain by using \eqref{eqn:QCoh[X/G]} an explicit model
\begin{flalign}\label{eqn:Fonobjects}
\FFF(V)\,:=\,\QCoh\big(\Sol(V)\big)\,\simeq\,{}_{\O(Z(V))}\dgMod_{\mathrm{cof}}^{\O(\G(V))}\,\in\,\dgCat
\end{flalign}
for the dg-category assigned to $V\in\P_{\bbZ^2}$ in terms of cofibrant 
$\O\big(Z(V)\big)$-dg-modules $M$ with a compatible $\O\big(\G(V)\big)$-coaction
$\rho_M : M\to M\otimes \O\big(\G(V)\big)$.
Given any $1$-ary operation $\iota_{V}^{V^\prime}: V\to V^\prime$ in $\P_{\bbZ^2}$, 
we obtain the dg-functor
\begin{flalign} \label{eqn:Fon1ary}
\FFF\big(\iota_V^{V^\prime}\big)\, :=\, \QCoh\big(\Sol\big(\iota_V^{V^\prime}\big)\big)\,:\,
\FFF(V)~\longrightarrow~\FFF(V^\prime)
\end{flalign}
which admits the following explicit description:
To an object $(M,\rho_M)$ in $\QCoh\big(\Sol(V)\big)$, it 
assigns the object in $\QCoh\big(\Sol(V^\prime)\big)$ which consists of
the cofibrant $\O\big(Z(V^\prime)\big)$-dg-module 
$\O\big(Z(V^\prime)\big)\otimes_{\O(Z(V))}M$
obtained by a change-of-base along the inclusion $\O\big(Z(V)\big)\subseteq \O\big(Z(V^\prime)\big)$
and the tensor product $\O\big(\G(V^\prime)\big)$-coaction
associated with $\rho : \O\big(Z(V^\prime)\big)\to \O\big(Z(V^\prime)\big)\otimes \O\big(\G(V^\prime)\big)$
and $\rho_M : M\to M\otimes\O\big(\G(V)\big)\to M\otimes \O\big(\G(V^\prime)\big)$, where
the last step uses the inclusion $\O\big(\G(V)\big)\subseteq \O\big(\G(V^\prime)\big)$.
On morphisms, the dg-functor $\FFF\big(\iota_V^{V^\prime}\big)$ is given
by change-of-base along the inclusion $\O\big(Z(V)\big)\subseteq \O\big(Z(V^\prime)\big)$,
which preserves the coaction equivariance properties of the 
hom-complexes.
\sk

It remains to define the prefactorization algebra $\FFF : \P_{\bbZ^2}\to \dgCat$
on operations of arity $0$ and $\geq 2$ in $\P_{\bbZ^2}$.
For an arity $0$ operation $\iota_{\varnothing}^V:\varnothing \to V$ in $\P_{\bbZ^2}$, 
this amounts to defining a dg-functor
\begin{flalign}\label{eqn:Fon0ary}
\FFF\big(\iota_{\varnothing}^{V}\big)\,:\,\FFF(\varnothing) \,=\,\mathsf{B}\bbK~\longrightarrow~\FFF(V) 
\end{flalign}
from the dg-category $\mathsf{B}\bbK\in\dgCat$ (the monoidal unit of $\dgCat$)
which consists of a single object with hom-complex $\bbK$. This datum is equivalent to
picking an object in $\FFF(V)$, for which we take 
the rank $1$ free $\O\big(Z(V)\big)$-dg-module with its given coaction 
$\rho : \O\big(Z(V)\big)\to \O\big(Z(V)\big)\otimes \O\big(\G(V)\big)$.
Given any $(n\geq 2)$-ary operation $\iota_{\und{V}}^V : \und{V} = (V_1,\dots,V_n)\to V$
in $\P_{\bbZ^2}$, we define the dg-functor
\begin{flalign}\label{eqn:Fonnary}
\FFF\big(\iota_{\und{V}}^V\big)\,:\,\bigotimes_{i=1}^n \FFF(V_i)~\longrightarrow~\FFF(V)
\end{flalign}
by the following construction: Recall that an object $\big((M_1,\rho_{M_1}),\dots,(M_n,\rho_{M_n})\big)$
in $\bigotimes_{i=1}^n \FFF(V_i)$ is a tuple of objects $(M_i,\rho_{M_i}) \in \FFF(V_i)$, for all $i=1,\dots,n$.
We endow the tensor product $\bigotimes_{i=1}^n M_i$ 
over $\bbK$ of the underlying $\O\big(Z(V_i)\big)$-dg-modules $M_i$
with the evident $\bigotimes_{i=1}^n \O\big(Z(V_i)\big)$-dg-module structure
and perform a change-of-base along the inclusion $\bigotimes_{i=1}^n \O\big(Z(V_i)\big)\subseteq \O\big(Z(V)\big)$.
This defines an $\O\big(Z(V)\big)$-dg-module which we endow with the 
tensor product $\O\big(\G(V)\big)$-coaction 
associated with $\rho : \O\big(Z(V)\big)\to \O\big(Z(V)\big)\otimes \O\big(\G(V)\big)$
and $\rho_{M_i} : M_i\to M_i\otimes\O\big(\G(V_i)\big)\to M_i\otimes \O\big(\G(V)\big)$, 
for all $i=1,\dots,n$, where the last step uses the inclusions 
$\O\big(\G(V_i)\big)\subseteq \O\big(\G(V)\big)$. This defines the object in
$\FFF(V)$ which is assigned by the dg-functor $\FFF\big(\iota_{\und{V}}^V\big)$ 
to $\big((M_1,\rho_{M_1}),\dots,(M_n,\rho_{M_n})\big)$.
On morphisms, the dg-functor $\FFF\big(\iota_{\und{V}}^{V}\big)$ is given
by taking tensor products over $\bbK$ and a
change-of-base along the inclusion 
$\bigotimes_{i=1}^n \O\big(Z(V_i)\big)\subseteq \O\big(Z(V)\big)$,
which preserves the coaction equivariance properties of the 
hom-complexes.
\begin{theo}\label{theo:PFA}
The construction above defines a $\dgCat$-valued 
prefactorization algebra $\FFF : \P_{\bbZ^2}\to \dgCat$
on the square lattice $\bbZ^2$. This prefactorization algebra is 
locally constant in the sense that $\FFF\big(\iota_V^{V^\prime}\big) : \FFF(V)\to\FFF(V^\prime)$
is a weak equivalence of dg-categories \cite{Tabuada}, for every $1$-ary operation
$\iota_V^{V^\prime} : V\to V^\prime$ in $\P_{\bbZ^2}$.
\end{theo}
\begin{proof}
Pseudo-multifunctoriality of the assignment 
$\FFF : \P_{\bbZ^2}\to \dgCat$ defined in \eqref{eqn:Fonobjects},
\eqref{eqn:Fon1ary}, \eqref{eqn:Fon0ary} and \eqref{eqn:Fonnary}
is a consequence of standard properties of tensor
products of dg-modules, which in particular imply 
pseudo-functoriality of change-of-base functors.
Local constancy follows from the result
in Theorem \ref{theo:locallyconstant} that 
$\Sol\big(\iota_V^{V^\prime}\big) : \Sol(V^\prime)\to \Sol(V)$
is a weak equivalence in $\dSt$, for every morphism 
$\iota_V^{V^\prime} : V\to  V^\prime$ in $\Rect(\bbZ^2)$. Hence,
$\FFF\big(\iota_V^{V^\prime}\big) = \QCoh\big(\Sol\big(\iota_V^{V^\prime}\big) \big): 
\FFF(V)=\QCoh\big(\Sol(V)\big)\to\FFF(V^\prime)= \QCoh\big(\Sol(V^\prime)\big)$ is a weak equivalence
in $\dgCat$ since the pseudo-functor $\QCoh$ in \eqref{eqn:QCohdSt} preserves weak equivalences.
\end{proof}

\begin{rem}\label{rem:quantization}
The prefactorization algebra $\FFF : \P_{\bbZ^2}\to \dgCat$ from
Theorem \ref{theo:PFA} describes a natural choice of classical observables for our lattice Yang-Mills model.
These observables are modeled in terms of the dg-categories of quasi-coherent
complexes $\FFF(V)=\QCoh\big(\Sol(V)\big)\in\dgCat$ on the local derived critical loci 
$\Sol(V)\in\dSt$ of this theory. It would be interesting to study and describe
dg-categorical deformation quantizations as in \cite{ToenQuant,CPTVV} 
of this prefactorization algebra or, to match better their setting, of
its variant $\FFF_{\mathrm{per}}$ which assigns the dg-subcategories 
$\FFF_{\mathrm{per}}(V) = \Perf\big(\Sol(V)\big)\subseteq \QCoh\big(\Sol(V)\big)$
of perfect complexes.
The natural input datum for such quantization constructions is given in our case
by the $(-1)$-shifted Poisson structure which is canonically 
defined on a derived critical locus. According
to the shifted deformation quantization philosophy, 
this should yield an ``$\mathbb{E}_{-1}$''-monoidal quantization
of the canonical symmetric monoidal structure on 
the dg-category $\FFF(V)$. While $\mathbb{E}_n$-monoidal quantizations have 
a concrete definition and interpretation for non-negative $n\geq 0$,
their meaning in the case of negative $n<0$ is unclear to us.
\sk

An alternative pathway towards quantizing
the prefactorization algebra $\FFF : \P_{\bbZ^2}\to \dgCat$ 
from Theorem \ref{theo:PFA} is given by leveraging 
Poisson additivity \cite{Safronov}
along one of the two dimensions of $\bbZ^2$ in order
to turn the $(-1)$-shifted Poisson structure into an unshifted one.
The associated deformation quantization problem then consists
of quantizing $\FFF(V)$ as an $\mathbb{E}_0$-monoidal dg-category,
i.e.\ a dg-category with a distinguished object.
Such $\mathbb{E}_0$-monoidal quantizations have
been worked out explicitly in simple examples, see \cite{BPSquantization}.
We expect that implementing Poisson additivity \cite{Safronov}
through the rather explicit homotopical Green's operator methods developed in \cite{BMS}
could lead to unshifted Poisson structures whose quantization can be described rather
concretely. A particularly interesting challenge arising in this quantization
problem is to understand from the perspective of prefactorization algebras
how $2$-dimensional quantum Yang-Mills theory develops its area-dependent 
features \cite{Witten}, which as shown in Theorem \ref{theo:PFA} 
are not present in the corresponding classical model.
We hope to come back to this issue in our future work.
\end{rem}

\begin{rem}\label{rem:nottoolittle}
Our concept of locally constant prefactorization algebras on square lattices $\bbZ^n$
is similar to the ``not too little disks'' algebras studied 
by Calaque and Carmona \cite{CalaqueCarmona}.
These algebras are encoded by an operad
which describes Euclidean disks in $\bbR^n$ of radius greater than
some fixed minimal radius $R_{\mathrm{min}}>0$
and their mutually disjoint inclusions into each other. Local constancy is encoded by
the $\infty$-localization of this operad at all $1$-ary operations. The main result of 
\cite{CalaqueCarmona} states that such ``not too little disks'' algebras are equivalent to the usual
$\mathbb{E}_n$-algebras, for any choice of minimal radius $R_{\mathrm{min}}>0$.
The authors also apply their techniques to study examples of $1$-dimensional lattice
prefactorization algebras on $\bbZ^1$, see also \cite{Calaquelattice} for further $1$-dimensional examples. 
We would like to add that the updated final version of \cite[Section 2.4]{CalaqueCarmona}
now contains a precise comparison between our $\P_{\bbZ^n}$-algebras
from Definition \ref{def:PFA} and ``not too little cubes'' algebras.
In particular, it is shown that in the locally constant case our prefactorization 
algebras on $\bbZ^n$ give rise to $\mathbb{E}_n$-algebras.
\end{rem}


\section*{Acknowledgments}
We would like to thank Victor Carmona and Alastair Grant-Stuart for useful discussions.
We also thank the anonymous reviewer for their meticulous report
which contained valuable comments and suggestions that helped us to improve our paper.
The work of M.B.\ is supported in part by the MUR Excellence 
Department Project awarded to Dipartimento di Matematica, 
Universit{\`a} di Genova (CUP D33C23001110001) and it is fostered by 
the National Group of Mathematical Physics (GNFM-INdAM (IT)). 
T.F.\ would like to thank ENS Lyon for financial support
and the University of Nottingham for their hospitality
during his internship project which has led to this work.
A.S.\ was supported by the Royal Society (UK) through a Royal Society University 
Research Fellowship (URF\textbackslash R\textbackslash 211015)
and Enhancement Grants (RF\textbackslash ERE\textbackslash 210053 and 
RF\textbackslash ERE\textbackslash 231077).


\appendix

\section{\label{app:details}Cohomology computation for Theorem \ref{theo:locallyconstant}}
In this appendix we prove that the inclusion
$A\subseteq \widetilde{A}^\prime$ of the
commutative algebra $A = \O\big(Z^{\gf}(V)\big)^0$ from \eqref{eqn:AandA'} 
into the commutative dg-algebra $\widetilde{A}^\prime\subseteq \O\big(Z^{\gf}(V^\prime)\big)$
from \eqref{eqn:widetildeA'} is a weak equivalence in $\dgCAlg^{\leq 0}$.
Since $A$ is discrete, this amounts to showing that this inclusion
induces an isomorphism 
$A \cong\mathsf{H}^0\big(\widetilde{A}^\prime\big)$ in $0$-th cohomology
and that the non-zero cohomologies $\mathsf{H}^{<0}\big(\widetilde{A}^\prime\big)=0$
of $\widetilde{A}^\prime\in\dgCAlg^{\leq 0}$ are trivial.
\sk

We start with proving the statement about the $0$-th cohomology.
From \eqref{eqn:AandA'} and \eqref{eqn:widetildeA'}, one observes 
that $\widetilde{A}^\prime$ is generated (non-freely) over $A$ 
by $T(x_1,d+1)$ and $\xi_1(x_1,d)$, for all $x_1\in[a,b-1]$.
Using the explicit form of the differential \eqref{eqn:OZYMtg}, 
we can write every degree $0$ generator as a sum
\begin{flalign}
\nn T(x_1,d+1)\,&=\,T(x_1,d)\,T(x_1,d-1)^{-1}\,T(x_1,d) \\
&\qquad - \dd\Big(T(x_1,d+1)\,\xi_1(x_1,d)\,T(x_1,d-1)^{-1}\,T(x_1,d)\Big)
\end{flalign}
of an element in $A$ and an exact term.
From this it follows that $A \cong\mathsf{H}^0\big(\widetilde{A}^\prime\big)$.
\sk

It remains to prove that the non-zero cohomologies
$\mathsf{H}^{<0}\big(\widetilde{A}^\prime\big)=0$ are trivial.
For this it is convenient to observe, by using the definition 
of $E^\gf(x)$ in \eqref{eqn:Efieldgf}, 
that the degree $0$ commutative subalgebra $\widetilde{A}^{\prime\, 0}$
of $\widetilde{A}^{\prime}$ in \eqref{eqn:widetildeA'} 
can be expressed equivalently in terms of the variables
$E^{\gf}(x) = T(x+e_2)^{-1}\,T(x)\in \O(\GL_n)$, for all $x\in[a,b-1]\times[c,d]$,
and $T(x_1,c)\in \O(\GL_n)$, for all $x_1\in[a,b-1]$.
The benefit of this change of variables is that the differential
$\dd\xi_{1}(x_1,d) = E^\gf(x_1,d)-E^\gf(x_1,d-1)$ 
of the degree $-1$ generators of $\widetilde{A}^{\prime}$ is linear in these variables,
for all $x_1\in[a,b-1]$.
Let us introduce the auxiliary commutative dg-algebra
\begin{flalign}\label{eqn:widetildeB'}
\widetilde{B}^{\prime}\,:=\, \bigotimes_{x\in[a,b-1]\times\{c\}}\O(\GL_n)\otimes \!\!\bigotimes_{x\in [a,b-1]\times[c,d]}\!\!
\O(\mathbb{A}^{n\times n})\otimes\!\! \bigotimes_{x\in [a,b-1]\times\{d\}}\!\! \Sym\big(\gl_n[1]\big)\quad,
\end{flalign}
where the first tensor factor describes $T(x_1,c)\in \O(\GL_n)$
and the second tensor factor describes $E^{\gf}(x)\in \O(\mathbb{A}^{n\times n})$
\textit{without} localization at the determinants $\det(E^{\gf}(x))$. The differential on 
$\widetilde{B}^{\prime}$ is defined by $\dd\xi_{1}(x_1,d) = E^\gf(x_1,d)-E^\gf(x_1,d-1)$,
for all $x_1\in[a,b-1]$.
There exists an evident $\dgCAlg^{\leq 0}$-morphism
$\widetilde{B}^{\prime}\to \widetilde{A}^{\prime}$ 
whose degree-zero component $\widetilde{B}^{\prime\,0}\to \widetilde{A}^{\prime\,0}$ 
is a localization of the commutative algebra $\widetilde{B}^{\prime\,0}$ at the determinants $\det(E^{\gf}(x))$,
for all $x\in[a,b-1]\times[c,d]$. Note that under change-of-base along this morphism we have an
isomorphism
\begin{flalign}\label{eqn:changeofbaseApp}
\widetilde{A}^{\prime\,0} \otimes_{\widetilde{B}^{\prime\,0}} \widetilde{B}^{\prime}\,\cong\, \widetilde{A}^\prime
\end{flalign}
of commutative dg-algebras.
\sk

We now observe that \eqref{eqn:widetildeB'} is a free extension of a discrete
commutative dg-algebra by the generators 
$\xi_{1}(x_1,d)$ and $\dd\xi_{1}(x_1,d)=E^\gf(x_1,d)-E^\gf(x_1,d-1)$,
for all $x_1\in[a,b-1]$, hence the non-zero cohomologies
$\mathsf{H}^{<0}\big(\widetilde{B}^\prime\big)=0$ are all trivial. 
(See e.g.\ \cite[Lemma 2.1]{CDGA} for a proof of this standard fact.)
Since localizations $\widetilde{B}^{\prime\,0}\to \widetilde{A}^{\prime\,0}$ of commutative algebras
are flat, it follows that the change-of-base
functor $\widetilde{A}^{\prime\,0}\otimes_{\widetilde{B}^{\prime\,0}}(-)$ is exact, hence
the isomorphism in \eqref{eqn:changeofbaseApp} implies that the 
non-zero cohomologies
$\mathsf{H}^{<0}\big(\widetilde{A}^\prime\big)=0$ of $\widetilde{A}^\prime$ are all trivial too.


\begin{thebibliography}{10}


\bibitem[AC22]{AnelCalaque}
M.~Anel and D.~Calaque,
``Shifted symplectic reduction of derived critical loci,''
Adv.\ Theor.\ Math.\ Phys.\ \textbf{26}, no.\ 6, 1543--1583 (2022)
[arXiv:2106.06625 [math.AG]].


\bibitem[AO21]{holim}
S.~Arkhipov and S.~Orsted,
``Homotopy limits in the category of dg-categories in terms of $A_\infty$-comodules,'' 
Eur.\ J.\ Math.\ \textbf{7}, no.\ 2, 671--705 (2021)
[arXiv:1812.03583 [math.CT]].


\bibitem[BMS23]{BMS}
M.~Benini, G.~Musante and A.~Schenkel,
``Green hyperbolic complexes on Lorentzian manifolds,''
Commun.\ Math.\ Phys.\ \textbf{403}, no.\ 2, 699--744 (2023)
[arXiv:2207.04069 [math-ph]].


\bibitem[BPSW21]{BPSW}
M.~Benini, M.~Perin, A.~Schenkel and L.~Woike, 
``Categorification of algebraic quantum field theories,'' 
Lett.\ Math.\ Phys.\ \textbf{111}, 35 (2021) 
[arXiv:2003.13713 [math-ph]].


\bibitem[BPS23]{BPSquantization}
M.~Benini, J.~P.~Pridham and A.~Schenkel,
``Quantization of derived cotangent stacks and gauge theory on directed graphs,''
Adv.\ Theor.\ Math.\ Phys.\ \textbf{27}, no.\ 5, 1275--1332 (2023)
[arXiv:2201.10225 [math-ph]].


\bibitem[BSS23]{BSSdCrit}
M.~Benini, P.~Safronov and A.~Schenkel,
``Classical BV formalism for group actions,''
Commun.\ Contemp.\ Math.\ \textbf{25}, no.\ 01, 2150094 (2023)
[arXiv:2104.14886 [math-ph]].


\bibitem[BS25]{BSchapter}
M.~Benini and A.~Schenkel,
``Operads, homotopy theory and higher categories in algebraic quantum field theory,''
in: R.~Szabo and M.~Bojowald (eds.), 
{\it Encyclopedia of Mathematical Physics}, 
Second Edition, Volume 5, 556--568 (2025) 
[arXiv:2305.03372 [math-ph]].



\bibitem[BBKK24]{Analytic}
O.~Ben-Bassat, J.~Kelly and K.~Kremnizer,
``A perspective on the foundations of derived analytic geometry,''
arXiv:2405.07936 [math.AG].


\bibitem[BZBJ18a]{Jordan1}
D.~Ben-Zvi, A.~Brochier and D.~Jordan, 
``Integrating quantum groups over surfaces,'' 
J.\ Topol.\ \textbf{11}, no.\ 4, 874--917 (2018)
[arXiv:1501.04652 [math.QA]].


\bibitem[BZBJ18b]{Jordan2}
D.~Ben-Zvi, A.~Brochier and D.~Jordan,
``Quantum character varieties and braided module categories,'' 
Sel.\ Math.\ New Ser.\ \textbf{24}, 4711--4748 (2018)
[arXiv:1606.04769 [math.QA]].


\bibitem[Ber14]{Berglund}
A.~Berglund,
``Homological perturbation theory for algebras over operads,''
Algebr.\ Geom.\ Topol.\ \textbf{14}, no.\ 5, 2511--2548 (2014)
[arXiv:0909.3485 [math.AT]].


\bibitem[Cal21]{Calaque}
D.~Calaque, 
``Derived stacks in symplectic geometry,'' 
in: M.~Anel and G.~Catren (eds.), 
{\it New Spaces in Physics: Formal and Conceptual Reflections}, 
155--201, Cambridge University Press, Cambridge (2021) 
[arXiv:1802.09643 [math.SG]].


\bibitem[Cal26]{Calaquelattice}
D.~Calaque,
``Not too little intervals for quantum mechanics,''
Port.\ Math.\ \textbf{83}, no.\ 1/2, 19--34 (2026)
[arXiv:2408.10033 [math.QA]].


\bibitem[CC26]{CalaqueCarmona}
D.~Calaque and V.~Carmona,
``Algebras over not too little discs,''
Commun.\ Math.\ Phys.\ \textbf{407}, 31 (2026)
[arXiv:2407.18192 [math.AT]].


\bibitem[CPTVV17]{CPTVV}
D.~Calaque, T.~Pantev, B.~To\"en, M.~Vaqui{\'e} and G.~Vezzosi,
``Shifted Poisson structures and deformation quantization,'' 
J.\ Topol.\ \textbf{10}, no.\ 2, 483 (2017) 
[arXiv:1506.03699 [math.AG]].


\bibitem[CG17]{CG1}
K.~Costello and O.~Gwilliam,
{\it Factorization algebras in quantum field theory: Volume 1},
New Mathematical Monographs \textbf{31}, 
Cambridge University Press, Cambridge (2017).


\bibitem[CG21]{CG2}
K.~Costello and O.~Gwilliam,
{\it Factorization algebras in quantum field theory: Volume 2},
New Mathematical Monographs \textbf{41}, 
Cambridge University Press, Cambridge (2021).


\bibitem[EP24]{Pridham}
J.~Eugster and J.~P.~Pridham,
``An introduction to derived (algebraic) geometry,''
Rend.\ Mat.\ Appl.\ (7) \textbf{45}(1-2), 1--109 (2024)
[arXiv:2109.14594 [math.AG]].


\bibitem[Gai15]{Gaitsgory}
D.~Gaitsgory,
``Sheaves of categories and the notion of $1$-affineness,''
in: T.~Pantev, C.~Simpson, B.~To\"en, M.~Vaqui{\'e} and G.~Vezzosi (eds.),
\textit{Stacks and categories in geometry, topology, and algebra}, 
Contemp.\ Math.\ \textbf{643}, American Mathematical Society, 
Providence, RI (2015)
[arXiv:1306.4304v4 [math.AG]].


\bibitem[GR17]{GRbook}
D.~Gaitsgory and N.~Rozenblyum,
{\it A study in derived algebraic geometry~I: Correspondences and duality},
Math.\ Surveys Monogr.\ \textbf{221},
American Mathematical Society, Providence, RI (2017).


\bibitem[GJ09]{Goerss}
P.~G.~Goerss and J.~F.~Jardine,
{\it Simplicial homotopy theory},
Modern Birkh\"auser Classics,
Birkh\"auser Verlag, Basel (2009).


\bibitem[Gra22]{Albin}
A.~Grataloup,
{\it Derived symplectic reduction and $\mathcal{L}$-equivariant geometry},
PhD thesis, Universit{\'e} de Montpellier (2022)
[arXiv:2302.04036 [math.AG]].


\bibitem[Kel06]{Keller}
B.~Keller, 
``On differential graded categories,'' 
International Congress of Mathematicians Vol.\ II, 151--190, 
Eur.\ Math.\ Soc., Z\"urich (2006) 
[arXiv:math/0601185 [math.KT]].


\bibitem[Kel82]{Kelly}
G.~M.~Kelly, 
{\it Basic concepts of enriched category theory}, 
London Math.\ Soc.\ Lecture Note Ser.\ \textbf{64}, 
Cambridge University Press, Cambridge-New York (1982). 
[Reprints in Theory and Applications of Categories 10, 1--136 (2005).]


\bibitem[KSY23]{DerivedPDE1}
J.~Kryczka, A.~Sheshmani and S.~Yau,
``Derived moduli spaces of nonlinear PDEs I: Singular propagations,''
arXiv:2312.05226 [math.AG].


\bibitem[KSY24]{DerivedPDE2}
J.~Kryczka, A.~Sheshmani and S.~Yau,
``Derived moduli spaces of nonlinear PDEs II: Variational tricomplex and BV formalism,''
arXiv:2406.16825 [math.AG].


\bibitem[LM20]{CDGA}
E.~Lepri and M.~Manetti,
``On deformations of diagrams of commutative algebras,'' 
in: E.~Colombo, B.~Fantechi, P.~Frediani, D.~Iacono and R.~Pardini (eds.), 
{\it Birational geometry and moduli spaces}, 
Springer INdAM Series \textbf{39}, Springer, Cham (2020)
[arXiv:1902.10436 [math.AG]].


\bibitem[LurHA]{Lurie}
J.~Lurie, 
{\it Higher algebra}, 
\url{https://www.math.ias.edu/~lurie/papers/HA.pdf}.


\bibitem[MM19]{Manetti}
M.~Manetti and F.~Meazzini,
``Formal deformation theory in left-proper model categories,''
New York J.\ Math.\ \textbf{25}, 1259--1311 (2019)
[arXiv:1802.06707 [math.CT]].


\bibitem[MM94]{Montvay}
I.~Montvay and G.~M\"unster,
{\it Quantum fields on a lattice},
Cambridge Monographs on Mathematical Physics,
Cambridge University Press, Cambridge (1994).


\bibitem[Rob12]{Roberts}
D.~M.~Roberts, 
``Internal categories, anafunctors and localisation,''
Theory Appl.\ Categ.\ \textbf{26}, no.\ 29, 788--829 (2012)
[arXiv:1101.2363 [math.CT]].


\bibitem[Saf18]{Safronov}
P.~Safronov,
``Braces and Poisson additivity,''
Compositio Math.\ \textbf{154}, 1698--1745 (2018)
[arXiv:1611.09668 [math.AG]].


\bibitem[Ste23]{Steffens1}
P.~Steffens,
``Derived $C^\infty$-geometry I: Foundations,''
arXiv:2304.08671 [math.AG].


\bibitem[Ste24]{Steffens2}
P.~Steffens,
``Representability of elliptic moduli problems in derived $C^\infty$-geometry,''
arXiv:2404.07931 [math.AG].


\bibitem[Tab05]{Tabuada}
G.~Tabuada,
``Une structure de cat{\'e}gorie de mod{\`e}les de Quillen sur la cat{\'e}gorie des dg-cat{\'e}gories,''
C.\ R.\ Math.\ Acad.\ Sci.\ Paris \textbf{340}, no.\ 1, 15--19 (2005)
[arXiv:math/0407338 [math.KT]].


\bibitem[To\"e07]{ToenDG}
B.~To\"en,
``The homotopy theory of dg-categories and derived Morita theory,''
Invent.\ Math.\ \textbf{167}, 615--667 (2007)
[arXiv:math/0408337 [math.AG]]. 


\bibitem[To\"e14a]{Toen}
B.~To\"en,
``Derived algebraic geometry,'' 
EMS Surv.\ Math.\ Sci.\ \textbf{1}, no.\ 2, 153--240 (2014)
[arXiv:1401.1044 [math.AG]].


\bibitem[To\"e14b]{ToenQuant}
B.~To\"en,
``Derived algebraic geometry and deformation quantization,'' 
Proceedings of the International Congress of Mathematicians, Seoul (2014) 
[arXiv:1403.6995 [math.AG]].


\bibitem[TV08]{ToenVezzosi}
B.~To\"en and G.~Vezzosi,
``Homotopical algebraic geometry II: Geometric stacks and applications,'' 
Mem.\ Amer.\ Math.\ Soc.\ \textbf{193}, no.\ 902 (2008) 
[arXiv:math/0404373 [math.AG]].


\bibitem[Vez20]{Vezzosi}
G.~Vezzosi, 
``Basic structures on derived critical loci,'' 
Differential Geom.\ Appl.\ \textbf{71}, 101635 (2020)
[arXiv:1109.5213 [math.AG]].


\bibitem[Wil74]{Wilson}
K.~G.~Wilson,
``Confinement of quarks,''
Phys.\ Rev.\ D \textbf{10}, 2445--2459 (1974).


\bibitem[Wit91]{Witten}
E.~Witten,
``On quantum gauge theories in two-dimensions,''
Commun.\ Math.\ Phys.\ \textbf{141}, 153--209 (1991).


\end{thebibliography}
\end{document}